\documentclass[11pt,letterpaper]{article}
\usepackage[margin=1in]{geometry}
\usepackage{ifpdf}
\ifpdf
    \usepackage[pdftex]{graphicx}
    \usepackage[update]{epstopdf}
\else
	\usepackage{graphicx}
\fi
\usepackage{wrapfig,bbm}
\usepackage{enumitem,color}
\usepackage{amsthm,multirow}
\newtheorem{theorem}{Theorem}
\newtheorem{lemma}{Lemma}
\newtheorem{corollary}{Corollary}

\newtheorem{prop}{\textbf{Proposition}}

\usepackage[mathscr]{euscript}
\usepackage{eufrak}
\usepackage{bbold}
\usepackage{array}
\usepackage{amsmath,amsthm}
\usepackage{amssymb}
\usepackage{amstext}
\usepackage{cite,setspace}
\usepackage[ruled,linesnumbered]{algorithm2e}

\newtheorem{definition}{\textbf{Definition}}

\renewcommand{\vec}[1]{{\boldsymbol{#1}}}

\DeclareGraphicsExtensions{.pdf,.png,.jpg,.eps}
\begin{document}

\title{Fundamental Limits of Coded Caching: From Uncoded Prefetching to Coded Prefetching}
\author{Kai Zhang and Chao Tian}
\maketitle

\begin{abstract}
In order to characterize the fundamental limit of the tradeoff between the amount of cache memory and the delivery transmission rate of multiuser caching systems, various coding schemes have been proposed in the literature. These schemes can largely be categorized into two classes, namely uncoded prefetching schemes and coded prefetching schemes. While uncoded prefetching schemes in general offer order-wise optimal performance, coded prefetching schemes often have better performance at the low cache memory regime. The significant differences in the coding components between the two classes may leave the impression that they are largely unrelated. In this work, we provide a connection between the uncoded prefetching scheme proposed by Maddah Ali and Niesen (and its improved version by Yu {\em et al.}) and the coded prefetching scheme proposed by Tian and Chen. A critical observation is first given where a coding component in the Tian-Chen scheme can be replaced by a binary code, which enables us to view the two schemes as the extremes of a more general scheme.  An explicit example is given to show that the intermediate operating points of this general scheme can in fact provide new memory-rate tradeoff points previously not known to be achievable in the literature. This new general coding scheme is then presented and analyzed rigorously, which yields a new inner bound to the memory-rate tradeoff for the caching problem. 
\end{abstract}

\section{Introduction}

Caching can be used to relieve contention on communication resources by prefetching data to a local or fast memory space, and thus avoiding data retrieval from the remote or slower data source during peak traffic time. Traditionally, caching has mainly been considered in single user settings, {\em e.g.,} on-CPU caches vs. RAM in computers, where the hit-ratio is the key measure of performance. As networked systems become more prevalent, caching systems involving multiple users have attracted increasingly more research attention. 

\begin{figure}[tb]
\centering
\includegraphics[width=8.5cm]{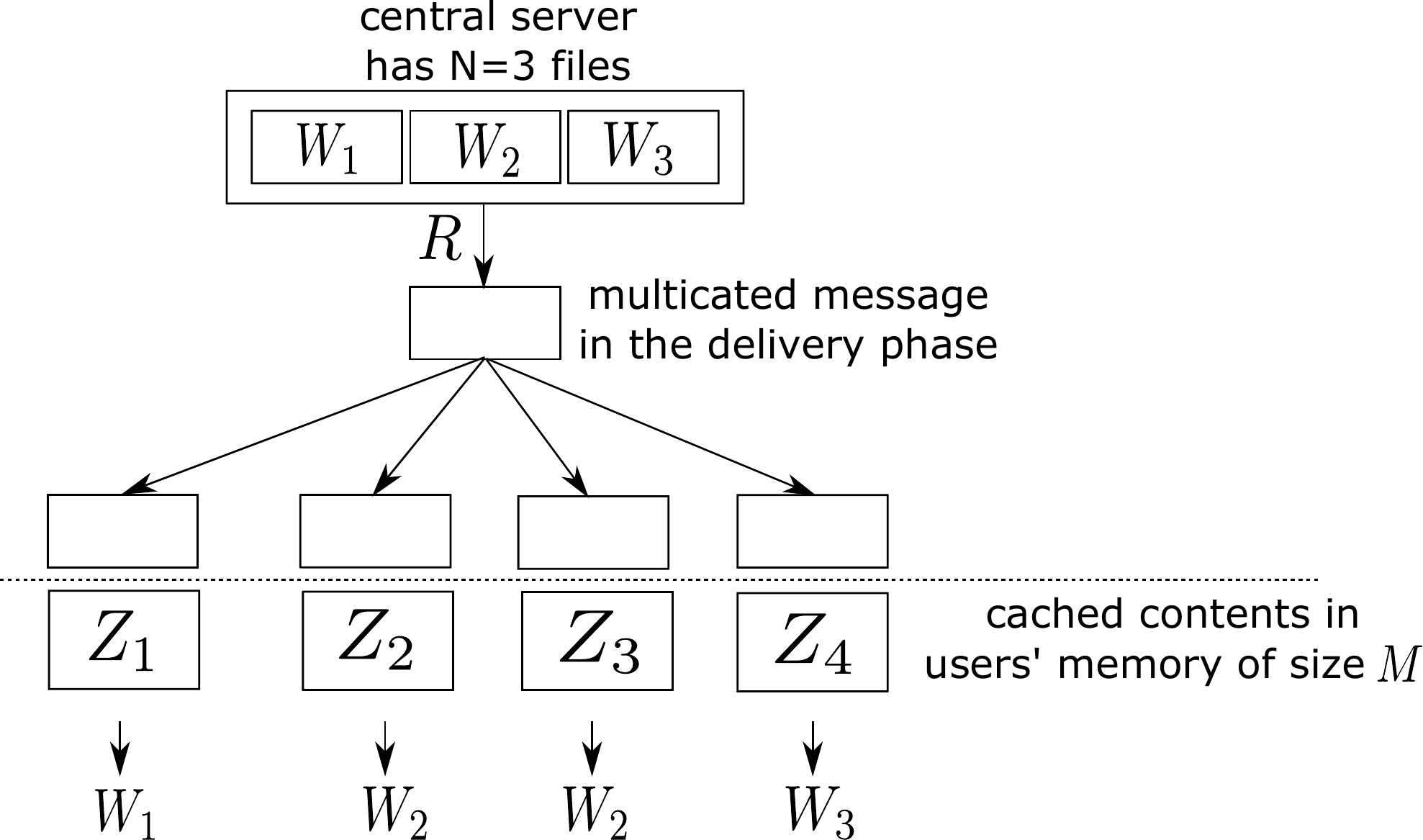}
\caption{\label{fig:system}An example caching system instance, where there are $N=3$ files, denoted as $(W_1,W_2,W_3)$, and $K=4$ users, whose cached contents are $(Z_1,Z_2,Z_3,Z_4)$, respectively. In this instance the users request files $(W_1,W_2,W_2,W_3)$, respectively.}
\end{figure}

In their award-winning article \cite{MaddahAliNiesen:14}, Maddah-Ali and Niesen provided a formal information theoretic formulation for the caching problem in multiuser settings; see Fig. \ref{fig:system}. In this formulation, there are $N$ files, each of $F$ bits, and $K$ users. Each user has a local cache memory of capacity $MF$ (thus a normalized capacity of $M$). In the prefetching  (or sometimes referred to as the placement) phase, the users can fill their caches with contents from the central server without the knowledge of the precise requests at the deliver phase. In the delivery phase, each user reveals the request for a single file from the central server, and the central server must multicast certain common (and possibly coded) information to all the users in order to accommodate these requests. Since in the prefetching phase, the requests at the later phase are unknown a-prior, the cached contents must be strategically prepared at all the users. The goal is to minimize the amount of multicast information which has rate $RF$ (or equivalently the normalized rate of {$R$)}, under the constraint on the normalized cache memory $M$. There is a natural tradeoff between the amount of cache memory and the delivery transmission rate, which is often referred to simply as the memory-rate tradeoff or the $(M,R)$ tradeoff. It was shown in \cite{MaddahAliNiesen:14} that in terms of this memory-rate tradeoff, coding can be rather beneficial, while solutions based on uncoded prefetching and delivery will suffer a significant loss. Subsequent works extended it also to decentralized caching placements \cite{MaddahAliNiesen:14Networking}, caching with nonuniform demands \cite{niesen2017coded}, online caching placements \cite{pedarsani2016online}, and hierarchical coded caching \cite{karamchandani2016hierarchical}, and many others.

There were quite a few recent efforts \cite{chen2016fundamental,Sahraei:15,Amiri:16,Wan:16, TianChen:16Arxiv,Yu:16,Amiri:17,gomez2016fundamental} aiming to find better codes with improved memory-rate tradeoff, toward the eventual goal of finding the optimal codes and thus completely characterizing the fundamental limit of this tradeoff. In particular, Yu {\em et al.} \cite{Yu:16} proposed a strategy that is optimal when prefetching is restricted to be uncoded, which in fact directly improves on the scheme in \cite{MaddahAliNiesen:14}. The key insight in \cite{Yu:16} appears to be that the original delivery strategy in \cite{MaddahAliNiesen:14} may have redundancy in the transmissions, which can be systematically removed to reduce the delivery rate in some cases. In another recent work, Tian and Chen \cite{TianChen:16Arxiv} proposed a coded prefetching and the corresponding delivery strategy, which relies on a combination of rank metric codes and maximum distance separable (MDS) codes in a non-binary finite field. In the regime when the memory size $M$ is relatively small, the scheme in \cite{TianChen:16Arxiv} can achieve a better performance than that in \cite{Yu:16}. Another code construction using coded prefetching was proposed more recently by G{\'o}mez-Vilardeb{\'o} in \cite{gomez2016fundamental}, which can provide further improvement, over the schemes in \cite{Yu:16} and \cite{TianChen:16Arxiv}, in the low memory regime for a specific range of $(N,K)$. The characterization of the fundamental limit of the memory-rate tradeoff however remains open, which appears to require both improved coding schemes and  stronger outer bounding techniques \cite{Tian:16Computer,ghasemi2017improved,wang2017improved,yu2017characterizing}.

In this work, we show that the scheme in \cite{TianChen:16Arxiv} can be slightly modified, where the MDS code used in the delivery phase can be replaced by a code using only binary additions (XOR). Though the alternative perspective itself does not provide further improvement on the known memory-rate tradeoff, it allows us to make a conceptual connection between the scheme in \cite{TianChen:16Arxiv} and that in \cite{Yu:16}. It further enables us to view these two schemes as the extremes of a more general scheme. The intermediate operating points of this more general scheme can indeed provide new tradeoff points previously not known in the literature, which we demonstrate using an explicit example for $(N,K)=(3,4)$. Extending this example, a general explicit code construction is then rigorously presented and analyzed, which provides a new inner bound to the fundamental limit of the memory-rate tradeoff region in the caching problem. The inner bound does not have a simple closed form expression, but can be represented as a linear program to facilitate its computation. 

The rest of the paper is organized as follows. Section \ref{sec:pre} first provides some necessary background on existing results and rank metric codes, and then introduces the notion of transmission type. A critical observation is given in Section \ref{sec:critical} which connects the two classes of schemes as extreme cases, and then a new memory-rate pair previous unknown in the literature is produced by considering the intermediate cases for $(N,K)=(3,4)$. 
The new inner bound is given formally in Section \ref{sec:alternative}, and the corresponding coding scheme, its analysis, the proof of the correctness are given in Section \ref{sec:newcode}.  Section \ref{sec:conclusion} concludes the paper with a few remarks. A technical proof is relegated to the appendix. 

\section{Relevant Results and Preliminaries}
\label{sec:pre}

In this section, we first briefly review existing results on the coded caching problem, and then provide necessary background on rank metric codes, which serve an instrumental role in the new code construction. A new concept important in our code construction, {\em i.e., } the transmission type, is then introduced. 

\subsection{Existing Schemes Using Uncoded Prefetching}
\label{subsec:codeexample1}

The scheme in \cite{MaddahAliNiesen:14}, which uses uncoded prefetching, can achieve the following memory-rate pairs
\begin{align}
(M,R)=\left(\frac{tN}{K},\frac{K-t}{1+t}\right),\quad t=0,1,\ldots,K,\label{eqn:MaddahAliNiesen:14}
\end{align}
and since another trivial point is clearly $(M,R)=(0,N)$, the lower convex hull of them provides an upper bound to the optimal tradeoff, as stated in \cite{MaddahAliNiesen:14}. 
More recently, Yu {\em et al.} \cite{Yu:16} gave a scheme which achieves the memory-rate tradeoff points of 
\begin{align}
(M,R)=\left(\frac{tN}{K},\frac{{K \choose t+1}-{{K-\min\{K,N\}} \choose t+1}}{{K \choose t}}\right),\quad t=0,1,\ldots,K.\label{eqn:Yu}
\end{align}
These points strictly improve the rate component $R$ in (\ref{eqn:MaddahAliNiesen:14}) when $K-N\geq t+1$. Both schemes in \cite{MaddahAliNiesen:14} and \cite{Yu:16} use the same uncoded prefetching strategy, but the delivery strategy in \cite{Yu:16} is a direct improvement to that in \cite{MaddahAliNiesen:14}. It was shown in \cite{Yu:16} that in the restricted class of schemes where only uncoded prefetching is allowed, the tradeoff provided in (\ref{eqn:Yu}) is in fact optimal. These two coding schemes can roughly be understood as follows.

Choose a fixed integer $t$, where $0\leq t\leq K$, and partition each file into ${K \choose t}$ segments of equal size; each segment is thus uniquely associated with a cardinality-$t$ subset $\mathcal{S}$ of the full user set $\{1,2,\ldots,K\}$, and this segment is placed in the caches of users in $\mathcal{S}$ during the prefetching phase. During the delivery phase, consider each $(t+1)$ subset $\mathcal{B}$ of users: within this group, each user is requesting a segment that is in all the other users' caches, and the server thus sends the XOR of all such segments of this group. Each user in this group can recover their respectively desired segment, since all other segments involved in this transmission are known to this user. As mentioned earlier, these transmissions as a whole, taken over all the possible choices of $\mathcal{B}$, may in fact have redundancy among themselves ({\em i.e.}, they are linearly dependent)  for certain $(N,K)$ parameters, and eliminating such redundancy results in the scheme in \cite{Yu:16}. 

Let us examine an example with $N=3$ files, denoted as $A,B,C$, respectively, and $K=4$ users. Set the auxiliary variable $t=2$, then each file is partitioned into ${4 \choose 2}=6$ segments, for example, file $A$ has segments  $A_{1,2},A_{1,3},A_{1,4},A_{2,3},A_{2,4},A_{3,4}$, and the segment $A_{1,2}$ is given to users 1 and 2, etc.. 
Suppose now the users' requests are $(A,A,B,C)$, {\em i.e.,} the first two users request file $A$, the third user requests file $B$, and the fourth user requests file $C$.  Consider the set of users $\mathcal{B}=\{1,2,3\}$, associated with which we should transmit $A_{2,3}+ A_{1,3}+ B_{1,2}$ according to the coding scheme discussed above, where the addition is in the binary field. Clearly the three users in $\mathcal{B}$ can recover their individually desired segments, {\em i.e.,} $A_{2,3},A_{1,3},B_{1,2}$, respectively. For other subsets of users, the transmissions are formed similarly, and the complete set of delivery transmissions is
\begin{align}
A_{2,3}+A_{1,3} + B_{1,2},A_{2,4}+ A_{1,4} + C_{1,2}, A_{3,4}+ B_{1,4}+ C_{1,3},A_{3,4}+ B_{2,4} + C_{2,3}.\label{eqn:34uncoded}
\end{align}
 In this particular case, the transmissions of (\ref{eqn:34uncoded}) do not have any redundancy.

The schemes in both \cite{Amiri:16} and \cite{Wan:16} use uncoded prefetching, and since the scheme in \cite{Yu:16} is optimal for this class of codes, the two schemes in  \cite{Amiri:16} and \cite{Wan:16} do not provide any additional improvement over  (\ref{eqn:Yu}).

\subsection{Existing Schemes Using Coded Prefetching}
\label{subsec:codedexample2}

Even in the pioneering work \cite{MaddahAliNiesen:14}, it was observed that uncoded  prefetching schemes are not sufficient to characterize the fundamental limit of the memory-rate tradeoff, and one code example using coded prefetching was given for the case $(N,K)=(2,2)$ as an illustration. In \cite{chen2016fundamental}, Chen {\em et al.} extended this example to the general case $N\leq K$, and showed that the single memory-rate pair $\left(\frac{1}{K},\frac{N(K-1)}{K}\right)$ is achievable and in fact optimal.

More recently, Tian and Chen \cite{TianChen:16Arxiv} proposed a more general scheme with coded prefetching for $N\leq K$. It was shown that the scheme can achieve the memory-rate tradeoff pairs
\begin{align}
(M,R)=\left(\frac{t[(N-1)t+K-N]}{K(K-1)},\frac{N(K-t)}{K}\right), \quad t=0,1,\ldots,K.\label{eqn:TianChen}
\end{align}
With $t=1$, it produces exactly the memory-rate pair given in \cite{chen2016fundamental}. 

The general scheme in \cite{TianChen:16Arxiv} is somewhat involved, but the digest is as follows. Each file is again partitioned into ${K \choose t}$ segments of equal size, and given to the relevant users as in \cite{Yu:16}; however, instead of directly storing them, each user caches certain linear combinations of these corresponding segments, mixed across all the files. During delivery, each symbol being transmitted is a linear combination of the segments from a single file, that serves two roles: firstly, the segments forming a single linear combination being transmitted are all present at certain user's cache that is not requesting this file, thus this user can use it to help resolve the cached symbols when sufficient such transmissions are collected; secondly, these segments are not present at some users which are requesting that file, thus can also help them to recover the missing segments. In order to guarantee the decodability, the cached contents and the transmitted contents should be made linearly independent, and for this purpose, rank metric codes can be utilized to produce the cached linear combinations, and MDS codes can be used to produce the delivery transmissions. 

Let us consider again the example  $(N,K)=(3,4)$ and $t=2$. In this case, the linear combinations of the segments 
\begin{align}
A_{1,2},A_{1,3},A_{1,4},B_{1,2},B_{1,3},B_{1,4},C_{1,2},C_{1,3},C_{1,4}\label{eqn:34_user1symbols}
\end{align}
are placed at user 1's cache, where each segment is viewed as a symbol in a large finite field. According to the scheme in \cite{TianChen:16Arxiv}, there should be a total of 5 linear combinations cached; the coefficients of these linear combinations are not critical in this construction, for which either deterministic rank metric codes can be used, or random assignments can be used with a high probability of being a valid choice in a sufficiently large finite field (which implies the existence of a deterministic assignment). Now consider again the requests $(A,A,B,C)$. In this case, the server will send the following 9 symbols
\begin{gather}
A_{3,4},B_{1,2},B_{1,4},B_{2,4},C_{1,2},C_{1,3},C_{2,3},A_{1,3}+A_{2,3},A_{1,4}+A_{2,4},\label{eqn:34coded}
\end{gather}
where the addition is in the same finite field of the information symbol which is usually not binary. The last two linear combinations can also be viewed as the parity symbols of two MDS codes. Now user 1 collects from (\ref{eqn:34coded}) the symbols $B_{1,2},B_{1,4},C_{1,2},C_{1,3}$, which, together with 5 cached linear combinations, leads to a total of $9$ linear combinations of the basis in (\ref{eqn:34_user1symbols}). Since the linear combinations are designed to be linearly independent, all of the symbols can be resolved. User 1 then collects $A_{1,3}+A_{2,3},A_{1,4}+A_{2,4}$ from which $A_{2,3}$ and $A_{2,4}$ can be recovered by eliminating $A_{1,3}$ and $A_{1,4}$, since they have been resolved from the cached content. It can be verified in a similar manner that all other users can also recover the requested files, and for any other demand patterns, transmissions of 9 symbols will always suffice.  The memory-rate pair achieved by the scheme in \cite{MaddahAliNiesen:14} is $(M,R)=(\frac{3}{2},\frac{2}{3})$ while the scheme in \cite{TianChen:16Arxiv} gives $(M,R)=(\frac{5}{6},\frac{3}{2})$, which are illustrated in Fig. \ref{fig:34}.

Amiri and Gunduz \cite{Amiri:17} showed that the following tradeoff point is achievable when $N\leq K$
\begin{align}
(M,R)=\left(\frac{N-1}{K},\frac{N(2K-N)}{2K}\right),\label{eqn:gunduz}
\end{align}
using a coded prefetching scheme. However, it can be verified that the pair $(M,R)$ in (\ref{eqn:gunduz}) is precisely on the time-sharing line between (\ref{eqn:Yu}) and (\ref{eqn:TianChen}) with $t=1$. More recently, G{\'o}mez-Vilardeb{\'o} \cite{gomez2016fundamental} showed that the following memory-rate pairs are achievable:
\begin{align}
(M,R)=\left(\frac{N}{Kg},N-\frac{N(N+1)}{K(g+1)}\right),\quad g=1,...,N, \label{eqn:gomez}
\end{align}
which can offer further improvement when $N\leq K\leq (N^2+1)/2$. The lower convex hull of (\ref{eqn:Yu}), (\ref{eqn:TianChen}) and (\ref{eqn:gomez}) provides the best known upper bound to the fundamental limit of $(M,R)$ tradeoff known in the literature.

\subsection{Linearized Polynomial and Rank Metric Codes}
\label{subsec:pre}
Similar as in \cite{TianChen:16Arxiv}, rank metric codes based on linearized polynomials (see \cite{Gab85}) can be used to facilitate our code constructions. The following lemma is relevant in this regard; see, {\em e.g.,} \cite{LidNie}.

\begin{lemma} \label{lem:useful}
A linearized polynomial in finite field $\mathbb{F}_{q^m}$
\begin{align}
f(x)=\sum_{i=1}^{P}v_i x^{q^{i-1}}, \  v_i \in \mathbb{F}_{q^m}
\label{eqn:linearized}
\end{align}
can be uniquely identified from evaluations at any $P$ points $x=\theta_i\in \mathbb{F}_{q^m}$, $i=1,2,\ldots,P$, that are linearly independent
over $\mathbb{F}_q$.
\end{lemma}

Another relevant property of linearized polynomials is that they satisfy the following condition
\begin{align}
f(ax + by)  =  af(x) + bf(y), \ a, b \in \mathbb{F}_q, \ x, y \in \mathbb{F}_{q^m},
\end{align}
which is the reason that they are called \lq\lq{}linearized\rq\rq{}. This property implies the following lemma, the proof of which can be found in \cite{TianChen:16Arxiv}.

\begin{lemma}
\label{lemma:fullrank}
Let $f(x)$ be a linearized polynomial in $\mathbb{F}_{q^m}$ as given in (\ref{eqn:linearized}), and let $\theta_i\in \mathbb{F}_{q^m}$, $i=1,2,\ldots,{P_o}$, be linearly independent over $\mathbb{F}_q$. Let $G$ be a $P_o\times P$ full rank (rank $P$) matrix with entries in $\mathbb{F}_q$, then $f(x)$ can be uniquely identified from
\begin{align}
[f(\theta_1),f(\theta_2),\ldots,f(\theta_{P_o})]\cdot G.
\end{align}
\end{lemma}

With a fixed set of $\theta_i\in \mathbb{F}_{q^m}$, $i=1,2,\ldots,{P_o}$, which are linearly independent,
we can view $(v_1,\ldots,v_P)$ as information symbols to be encoded, and the evaluations $[f(\theta_1),f(\theta_2),\ldots,f(\theta_{P_o})$] as the coded symbols. This is a $(P_o,P)$ MDS code in terms of rank metric \cite{Gab85}, where $P_o\geq P$. More importantly, the above lemma says any full rank (rank $P$) $\mathbb{F}_q$ linear combinations of the coded  symbols are sufficient to decode all the information symbols. This linear-transformation-invariant property had been utilized previously in other coding problems such as network coding with errors and erasures \cite{Koetter:08}, locally repairable codes with regeneration \cite{Silberstein:15}, and layered regenerating codes \cite{Tian:15}.

The codes thus obtained are not systematic, but they can be converted to systematic codes by viewing the information symbols $(w_1,w_2,\ldots,w_P)$ as the first $P$ evaluations $[f(\theta_1),f(\theta_2),\ldots,f(\theta_{P})]$, which can be used to find the coefficients of the linearized polynomial $(v_1,v_2,\ldots,v_P)$, and then the additional parity symbols can be generated by evaluating this linearized polynomial at the remaining points $(\theta_{P+1},\ldots,\theta_{P_o})$. Systematic rank-metric codes are instrumental in our construction.

\subsection{Demand Vectors and Transmission Types}

Denote the $N$ files in the system as $W_1,W_2,\ldots,W_N$, and denote the demands by the users in the delivery phase as $\vec{d}=(d_1,d_2,\ldots,d_K)$, where $d_k\in \{1,2,...,N\}$ is the index of the file that user-$k$ requests. For convenience, denote the set $\{1,2,\ldots,n\}$ as $I_n$. Recall that once the auxiliary parameter $t$ is fixed, each file $W_n$ in the scheme of \cite{MaddahAliNiesen:14} is the collection of all segments $W_{n,\mathcal{S}}$ where $\mathcal{S}\subseteq I_K$ and $|\mathcal{S}|=t$, where $|\mathcal{S}|$ is the cardinality of the set $\mathcal{S}$. For a given demand vector $\vec{d}=(d_1,d_2,\ldots,d_K)$, denote the set of users requesting file $W_n$ as
\begin{align}
&I^{[n]}\triangleq \{k\in I_K: \text{user $k$ requests file } W_n\},\qquad n=1,2,\ldots,N.
\end{align}
Further define $m_n\triangleq|I^{[n]}|$, $n=1,2,\ldots,N$. In the coding scheme we shall present, an arbitrary element (for example, the minimum element) in $I^{[n]}$ will be chosen, denoted as $\ell^{[n]}$, as the leader of $I^{[n]}$. The support of vector $\vec{m}$ is written as $\mathbb{supp}(\vec{m})$, {\em i.e.}, $\mathbb{supp}(\vec{m})=\{n| m_n>0\}$, and its cardinality is denoted as $N^*=|\mathbb{supp}(\vec{m})|$, which is the number of files being requested in $\vec{d}$. For convenience, also define $\tilde{N}\triangleq \min(N,K)$.

The notion of the {\em transmission type} is associated with each transmission in the scheme in \cite{MaddahAliNiesen:14}. For a set of users $\mathcal{B}\subseteq I_K$ where $|\mathcal{B}|=t+1$, the associated delivery transmission in the scheme of \cite{MaddahAliNiesen:14}, for a fixed demand vector $(d_1,d_2,\ldots,d_K)$, can be compactly written as the binary field summation
\begin{align}
\oplus_{k\in \mathcal{B}} W_{{d_k},\mathcal{B}\setminus k}. \label{eqn:compact}
\end{align}
Each such transmission, or alternatively the subset $\mathcal{B}$, is thus associated with an $N$-dimensional vector $\vec{t}$, whose $n$-th coordinate $t_n$ specifies the number of users that are demanding file $W_n$ in the set $\mathcal{B}$. We call this vector the transmission type of the subset $\mathcal{B}$. For example, in the $(3,4)$ case discussed above when the demand vector is $(W_1,W_2,W_3,W_4)=(A,A,B,C)$, the transmission type of the user set $\mathcal{B}=\{1,2,3\}$ is $\vec{t}=(2,1,0)$, and the exact transmission is $A_{2,3}+A_{1,3}+B_{1,2}$ where there are exactly two $W_1=A$ symbols involved and one $W_2=B$ symbol involved. Similarly, the transmission types of the user sets $\{2,3,4\}$ and $\{1,3,4\}$ are both $\vec{t}=(1,1,1)$. 

Denote the collection of all valid transmission types for a given demand vector $\vec{d}$ with the auxiliary parameter being $t$ as $\mathcal{T}^{(t)}_{\vec{d}}$. It is clear that for any valid transmission type $\vec{t}\in \mathcal{T}^{(t)}_{\vec{d}}$, we have $\sum_{n=1}^N t_n=t+1$, and thus the auxiliary parameter $t$ can be uniquely determined from any valid $\vec{t}$. The support of a transmission type $\vec{t}$ is denoted as $\mathbb{supp}(\vec{t})$. With a slight abuse of notation, we write the transmission type of a given set $\mathcal{B}\subseteq I_K$, where $|\mathcal{B}|=t+1$, as $\mathcal{T}(\mathcal{B})$. 

The notion of transmission type should be contrasted to the notion of {\em demand type} introduced in \cite{Tian:16symmetry}, which is a length-$N$ vector formed by sorting $(m_1,m_2,\ldots,m_N)$. This notion is also important in our work, because the symmetry in the proposed code implies that only one demand vector per demand type needs to be considered. Denote the collection of the representative demand vectors, one representative demand vector per demand type, as $\mathcal{D}$. 

\section{A Hidden Connection and Partial Decomposition}
\label{sec:critical}

The two schemes in \cite{TianChen:16Arxiv} and \cite{MaddahAliNiesen:14} (and its improved counterpart \cite{Yu:16}) may seem very different at the first sight:  one uses coded prefeteching and the other uncoded, one is non-binary code while the other is binary, and one relies on sophisticated coding techniques such as rank metric codes and the other only relatively simple combinatorics. Nevertheless, a closer look reveals some curious connections between the two schemes. For example, the tradeoff points in (\ref{eqn:TianChen}) lead to the rate values $R=\frac{N(K-t)}{K}$ for $t=0,1,\ldots,K$, which are exactly the same set of $M$ values given in (\ref{eqn:Yu}). This connection may or may not be a simple coincidence, however we next describe a much less obvious observation which leads to the main result of this paper.

\subsection{A Hidden Connection}

Consider again the example case for $(N,K)=(3,4)$ and $t=2$. Let us decompose the transmissions in (\ref{eqn:34uncoded}) by separating different files in the same linear combination. For example, the linear combination $A_{2,3}+ A_{1,3}+ B_{1,2}$ is decomposed into a pair of transmissions $(A_{2,3}+ A_{1,3},B_{1,2})$. It can be verified that decomposing all the linear combinations in (\ref{eqn:34uncoded}) in fact produces exactly the same set of linear combinations in (\ref{eqn:34coded}), after removing the repeated transmissions. Thus in this example, the delivery transmissions in the scheme \cite{TianChen:16Arxiv} can be obtained by fully decomposing the delivery transmissions of the scheme in \cite{MaddahAliNiesen:14}, when the auxiliary parameter $t$ is chosen to be the same in the two schemes. 

We note that the addition in (\ref{eqn:34coded}) is not in a binary field, while the addition in (\ref{eqn:34uncoded}) is  in the binary field. However, if a binary extension field $\mathbb{F}_{2^m}$ is used in (\ref{eqn:34coded}), the delivery can indeed be accomplished using only additions of the information symbols in this binary extension field, {\em i.e.,} the coefficients of the linear combinations are either $0$ or $1$. Clearly, such additions are equivalent to additions in the base binary field, when the information and coded symbols are represented in their binary vector form. 

For other $(N,K)$ parameters, by replacing the MDS code component in the coding scheme in \cite{TianChen:16Arxiv} with such decomposed delivery transmissions from the scheme of \cite{MaddahAliNiesen:14}, an alternative version of the code given in \cite{TianChen:16Arxiv} can be obtained. In fact, instead of simply presenting this alternative code construction, an even more general construction shall be provided, based on an example given next.

\begin{figure}[tb]
\centering
\includegraphics[width=12cm]{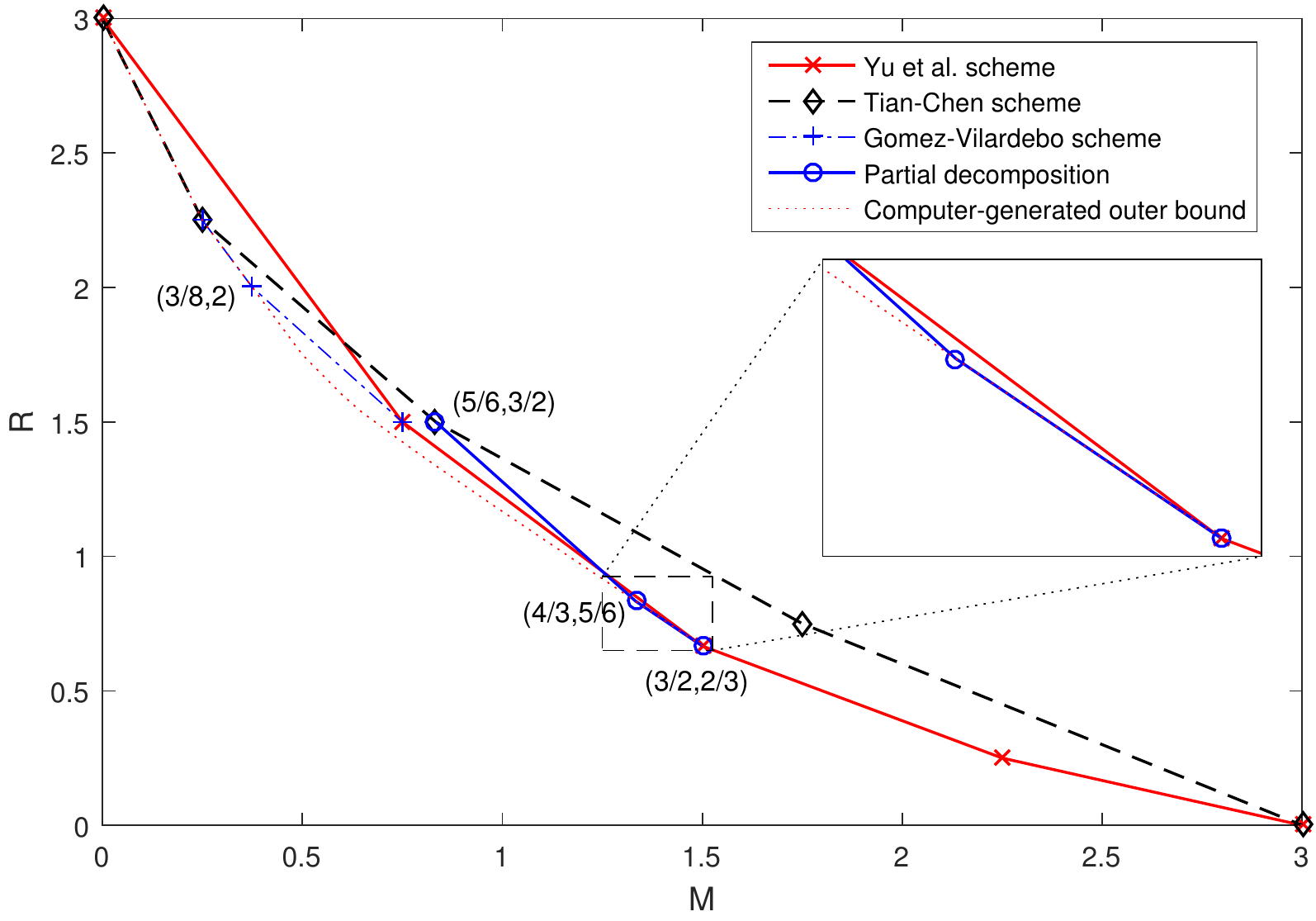}
\caption{A new tradeoff point for $(N,K)=(3,4)$. \label{fig:34}}
\end{figure}

\subsection{Partial Decomposition and a New Code Example}
\label{subsec:new}

The above observation naturally raises the following question: since the delivery strategy in the scheme of \cite{TianChen:16Arxiv} can be viewed as being obtained from fully decomposing the delivery transmissions of the scheme in \cite{MaddahAliNiesen:14}, will partial decomposition, with a correspondingly modified prefetching strategy, produce new memory-rate tradeoff pairs?  Next we provide an example code, which shows that the answer to this question is indeed positive. 

Consider again the case $(N,K)=(3,4)$ and $t=2$, but this time each user caches 8 (instead of 9 as in Sec. \ref{subsec:codeexample1}, or 5 as in Sec. \ref{subsec:codedexample2}) linear combinations of the information symbols of the corresponding uncoded file segments. The coefficients of the linear combinations can again be either from deterministic rank metric codes (see Section \ref{sec:newcode}), or generated randomly in a large finite field.

We next argue that delivering a total of $5$ coded symbols is sufficient in this case, which gives an achievable memory-rate pair $(M,R)=(4/3,5/6)$. The memory-rater pair is strictly better than $(4/3,23/27)$ achieved by the lower convex hull of the schemes \cite{Yu:16,TianChen:16Arxiv,gomez2016fundamental}, which is currently the best known upper bound in the literature; see Fig. \ref{fig:34} for an illustration. For completeness, a computer-generated outer bound is also included in the figure, which was obtained in a separate work \cite{Tian:16Computer}. Interestingly, both $(M,R)=(3/8,2)$ given by the code in \cite{gomez2016fundamental} and $(M,R)=(4/3,5/6)$ obtained in this work are in fact on this outer bound, and thus optimal. 

Due to the symmetry in the code, we only need to consider the demand vectors $(A,A,B,C)$, $(A,A,B,B)$, $(A,A,A,C)$, and $(A,A,A,A)$. 
\begin{itemize}
\item For the demand $(A,A,B,C)$, instead of fully decomposing the transmissions in (\ref{eqn:34uncoded}), we now partially decompose them as 
\begin{gather}
A_{2,3}+A_{1,3}+B_{1,2},A_{2,4}+A_{1,4}+C_{1,2},B_{1,4}+C_{1,3},B_{2,4}+C_{2,3},A_{3,4}.\label{eqn:partialdecompose}
\end{gather}
User 1 first collects $B_{1,4}+C_{1,3}$, and thus together with the 8 cached linear combinations, can resolve all the symbols in (\ref{eqn:34_user1symbols}) since he has a total of $9$ linearly-independent linear combinations of the 9 symbols; now user 1 essentially has uncoded cache contents, and thus can recover the needed file segments of $A$ (which are $A_{2,3},A_{2,4},A_{3,4}$) using the remaining transmissions. Users 2, 3, and 4 can use a similar strategy. 
\item For the demand $(A,A,B,B)$, we can transmit the following symbols
\begin{gather}
A_{2,3}+A_{1,3}+B_{1,2},A_{2,4}+A_{1,4}+B_{1,2},B_{1,4}+B_{1,3},B_{2,4}+B_{2,3},A_{3,4}.\label{eqn:partialdecompose1}
\end{gather}
User 1 can collect $B_{1,4}+B_{1,3}$ in order to resolve the cached symbols, and the decoding is similar to the previous case. It is also obvious user 3 and user 4 can indeed recover file $B$.
\item For the demand $(A,A,A,C)$, we can transmit
\begin{gather}
A_{2,3}+A_{1,3}+A_{1,2},A_{2,4}+A_{1,4}+C_{1,2},A_{1,4}+C_{1,3},A_{2,4}+C_{2,3},A_{3,4}.\label{eqn:partialdecompose2}
\end{gather}
The decoding strategies of other users are similar to the previous cases, and let us only consider user 3 as an illustration. User 3 can use $A_{3,4}$ to resolve the symbols in the cache, then recover $A_{1,4}$ from $A_{1,4}+C_{1,3}$, $A_{2,4}$ from $A_{2,4}+C_{2,3}$, and $A_{1,2}$ from $A_{2,3}+A_{1,3}+A_{1,2}$. 
\item For the demand $(A,A,A,A)$, we can transmit
\begin{gather}
A_{2,3}+A_{1,3}+A_{1,2},A_{2,4}+A_{1,4}+A_{1,2},A_{1,4}+A_{1,3},A_{2,4}+A_{2,3},A_{3,4}.\label{eqn:partialdecompose3}
\end{gather}
Using a similar strategy as above, it is seen that all users can indeed recover file $A$.
\end{itemize}

In this example case, the delivery transmissions are obtained by partially decomposing the transmissions in the scheme of \cite{MaddahAliNiesen:14}, and in compensation, the number of cached linear combinations in users' memory is reduced from that of \cite{MaddahAliNiesen:14}. The number of linear combinations stored in the cache needs to guarantee that the coded symbols can all be resolved to their uncoded form, after a sufficient number of symbols have been collected from the delivery. The rest of the paper is devoted to the task of using this idea to build a general class of codes which yield a new inner bound to the memory-rate tradeoff.

\section{A New Inner Bound  to the Optimal Memory-Rate Tradeoff}
\label{sec:alternative}

We first formally define the partial decomposition patterns, and then present the new inner bound. The prefetching strategy and the delivery strategy behind this new bound are presented and analyzed in the next section. 

\subsection{A Formal Description of Partial Decomposition}

Fix the auxiliary parameter $t\in I_K$, and for now also consider a fixed demand vector $\vec{d}$. A valid partial decomposition pattern on a transmission type $\vec{t}$ is specified by a partition $\mathcal{P}_{\vec{t},\vec{d}}$ on $\mathbb{supp}(\vec{t})$, {\em i.e.,} the elements of $\mathcal{P}_{\vec{t},\vec{d}}$ are mutually exclusive and jointly exhaustive subsets of $\mathbb{supp}(\vec{t})$. For a given transmission type $\vec{t}$ and its partial decomposition pattern $\mathcal{P}_{\vec{t},\vec{d}}$, the decomposed transmissions are formed by keeping the symbols in the same partition in $\mathcal{P}_{\vec{t},\vec{d}}$ together, but those across partitions separated. More precisely, let $\mathcal{T}(\mathcal{B})=\vec{t}$, then the transmission (\ref{eqn:compact}) can be rewritten and thus decomposed as 
\begin{align}
\oplus_{k\in \mathcal{B}} W_{{d_k},\mathcal{B}\setminus k}=\oplus_{\mathscr{P}\in \mathcal{P}_{\vec{t},\vec{d}}} \left[\oplus_{n\in\mathscr{P}}\left(\oplus_{k\in \mathcal{B}\cap I^{[n]}} W_{n,\mathcal{B}\setminus k}\right)\right]\quad\Rightarrow\quad \oplus_{n\in\mathscr{P}}\left(\oplus_{k\in \mathcal{B}\cap I^{[n]}} W_{n,\mathcal{B}\setminus k}\right),\,  \mathscr{P} \in \mathcal{P}_{\vec{t},\vec{d}}, \label{eqn:directdecomp}
\end{align}
where $\mathscr{P}\subseteq \mathbb{supp}(\vec{t})$ is used to enumerate over the partitions specified by $\mathcal{P}_{\vec{t},\vec{d}}$. Note that $\mathcal{P}_{\vec{t},\vec{d}}$ is the decomposition pattern for a transmission type $\vec{t}$ (and a demand vector $\vec{d}$), which implies that the transmissions of the same transmission type are not allowed to use different decomposition patterns. In order to specify the delivery transmissions for a given demand vector $\vec{d}$, the decomposition patterns for all transmission types should be given, which are written as a set $\vec{\mathcal{P}}^{(t)}_{\vec{d}}\triangleq \{\mathcal{P}_{\vec{t},\vec{d}}| \vec{t}\in \mathcal{T}^{(t)}_{\vec{d}}\}$. 

Consider again the example for $(N,K)=(3,4)$: suppose the demand vector is $$\vec{d}=(A,A,B,C)=(1,1,2,3),$$ and for the transmission type $\vec{t}=(1,1,1)$, the decomposition pattern is $\mathcal{P}_{(1,1,1),(1,1,2,3)}=\{\{1\},\{2,3\}\}.$ With these settings the two transmissions $A_{3,4}+B_{1,4}+C_{1,3}$ and $A_{3,4}+B_{2,4}+C_{2,3}$ in the coding scheme \cite{MaddahAliNiesen:14} will be decomposed into $
\{A_{3,4},B_{1,4}+C_{1,3}\}$ and $\{A_{3,4},B_{2,4}+C_{2,3}\}$, respectively.

For any demand vector $\vec{d}$, a special uncoded transmission pattern, denoted as $\breve{\vec{\mathcal{P}}}_{\vec{d}}^{(t)}$, is also allowed. When $K-t\geq \tilde{N}$, this strategy corresponds to directly transmitting a subset of files in the uncoded form. For general parameters, the transmission strategy will be given more precisely in Section \ref{sec:delivery}. The introduction of this pattern is motivated by the coding strategy in \cite{TianChen:16Arxiv} when $N^*<\tilde{N}$.

\subsection{A New Inner Bound}

Define the following quantity for any transmission pattern $\vec{\mathcal{P}}_{\vec{d}}^{(t)}$ except $\vec{\mathcal{P}}_{\vec{d}}^{(t)}=\breve{\vec{\mathcal{P}}}_{\vec{d}}^{(t)}$
\begin{equation}
R_{\vec{d},\vec{\mathcal{P}}^{(t)}_{\vec{d}}} \triangleq \sum_{\vec{t} \in \mathcal{T}^{(t)}_{\vec{d}}} \sum_{\mathscr{P} \in \mathcal{P}_{\vec{t},\vec{d}}} \left[\Bigg(\prod_{n \in \mathscr{P}} \binom{m_n}{t_n} - \prod_{n \in \mathscr{P}} \binom{m_n-1}{t_n} \Bigg)\cdot\prod_{n\in \mathbb{supp}(\vec{t})\setminus \mathscr{P}}{m_n \choose t_n}\right], \label{eqn:Rs}
\end{equation}
and for $k=1,2,\ldots,K$, 
\begin{align}
& M_{\vec{d}, \vec{\mathcal{P}}^{(t)}_{\vec{d}}, k} \triangleq  N \binom{K-1}{t-1} - \Delta M_{\vec{d},\vec{\mathcal{P}}^{(t)}_{\vec{d}}, k},\label{eqn:Ms}
\end{align}
where 
\begin{align}
&\Delta M_{\vec{d},\vec{\mathcal{P}}^{(t)}_{\vec{d}}, k}\triangleq \nonumber\\
&\quad\sum_{\substack{\vec{t} \in \mathcal{T}^{(t)}_{\vec{d}}:\\ t_{d_k} >0}}
\sum_{\substack{\mathscr{P} \in \mathcal{P}_{\vec{t},\vec{d}}:\\ d_k \notin \mathscr{P}}} \binom{m_{d_k}-1}{t_{d_k}-1}\left[\left( \prod_{n \in \mathscr{P}}\binom{m_n}{t_n}-\prod_{n \in \mathscr{P}} \binom{m_n-1}{t_n}\right)\cdot \prod_{n \in \mathbb{supp}(\vec{t})\setminus\{\mathscr{P}\cup\{d_k\}\}} \binom{m_n}{t_n}\right].\label{eqn:DeltaM}
\end{align}
For the special transmission pattern $\vec{\mathcal{P}}^{(t)}_{\vec{d}}=\breve{\vec{\mathcal{P}}}^{(t)}_{\vec{d}}$, the corresponding quantities are defined as
\begin{align}
R_{\vec{d},\breve{\vec{\mathcal{P}}}^{(t)}_{\vec{d}}} \triangleq \min(K-t,\tilde{N}){K \choose t}, \label{eqn:Rs1}
\end{align}
and for $k=1,2,\ldots,K$,
\begin{align}
& M_{\vec{d}, \breve{\vec{\mathcal{P}}}^{(t)}_{\vec{d}}, k} \triangleq  N \binom{K-1}{t-1} - \Delta M_{\vec{d},\breve{\vec{\mathcal{P}}}^{(t)}_{\vec{d}}, k}, \label{eqn:Ms1}
\end{align}
where
\begin{align}
\Delta M_{\vec{d},\breve{\vec{\mathcal{P}}}^{(t)}_{\vec{d}}, k}\triangleq\min(K-t,\tilde{N}) {K-1 \choose t-1}.
\end{align}
In the above, the following convention for the degenerate cases of combinatorics has been used
\begin{align}
\binom{a}{b} = 
\begin{cases}
    0, & \text{if }  a<b\\
    1, & \text{if } a \geq 0 \mbox{ and } b =0
\end{cases}.
\end{align}

Intuitively speaking, the vector $(M_{\vec{d}, \vec{\mathcal{P}}^{(t)}_{\vec{d}}, 1},...,M_{\vec{d}, \vec{\mathcal{P}}^{(t)}_{\vec{d}}, K},R_{\vec{d},\vec{\mathcal{P}}^{(t)}_{\vec{d}}})$ provides the cache memory requirements at the users and the rate requirement on the delivery transmission in the proposed coding scheme, when the demand vector $\vec{d}$ and the decomposition patterns $\vec{\mathcal{P}}^{(t)}_{\vec{d}}$ are fixed. Thus these numbers roughly provide the memory-rate tradeoff for the specific demand vector for a fixed decomposition pattern. 

We first observe that there may still be unbalance among the cache memory requirements at different users $(M_{\vec{d}, \vec{\mathcal{P}}^{(t)}_{\vec{d}}, 1},...,M_{\vec{d}, \vec{\mathcal{P}}^{(t)}_{\vec{d}}, K})$, meaning that different users may have different cache memory requirements under the decomposition pattern $\vec{\mathcal{P}}^{(t)}_{\vec{d}}$. This issue can be mitigated by coding across multiple instances in the prefetching phase and then producing the delivery transmissions using multiple different decomposition patterns on different instances to achieve better balance among users; note that simple space-sharing is not sufficient to achieve such a performance. A second important observation is that regardless the demand vectors or the decomposition patterns, the caching strategy can essentially be kept the same, which is to store a certain number of linear combinations of a fixed set of symbols. The two observations lead to the following definition and the main theorem below, the formal proof of which will be given in the sequel.

Define the region $\mathcal{R}^{(t)}$ to be the collection of the memory-rate pairs $(M,R)$ such that there exists a set of real-valued $\{\alpha_{\vec{d},\vec{\mathcal{P}}^{(t)}_{\vec{d}}}\}$ such that
\begin{align}
\sum_{\vec{\mathcal{P}}^{(t)}_{\vec{d}}} \alpha_{\vec{d},\vec{\mathcal{P}}^{(t)}_{\vec{d}}}&=1,&\quad \forall \vec{d}\in {\mathcal{D}}\label{eqn:LP1}\\
\alpha_{\vec{d},\vec{\mathcal{P}}^{(t)}_{\vec{d}}}&\geq 0,&\quad \forall \vec{d}\in{\mathcal{D}}, \forall \vec{\mathcal{P}}^{(t)}_{\vec{d}}\\
\alpha_{\vec{d},\vec{\mathcal{P}}^{(t)}_{\vec{d}}}&\leq 1,&\quad \forall \vec{d}\in{\mathcal{D}}, \forall \vec{\mathcal{P}}^{(t)}_{\vec{d}}\\
\sum_{\vec{\mathcal{P}}^{(t)}_{\vec{d}}} \alpha_{\vec{d},\vec{\mathcal{P}}^{(t)}_{\vec{d}}} R_{\vec{d},\vec{\mathcal{P}}^{(t)}_{\vec{d}}}&\leq R {K \choose t},&\quad \forall \vec{d}\in {\mathcal{D}}\\
\sum_{\vec{\mathcal{P}}^{(t)}_{\vec{d}}} \alpha_{\vec{d},\vec{\mathcal{P}}^{(t)}_{\vec{d}}} M_{\vec{d},\vec{\mathcal{P}}^{(t)}_{\vec{d}},k}&\leq M {K \choose t},&\quad \forall \vec{d}\in{\mathcal{D}}, \forall k\in I_K.\label{eqn:LP5}
\end{align}

The auxiliary variables $\{\alpha_{\vec{d},\vec{\mathcal{P}}^{(t)}_{\vec{d}}}\}$ serve a similar role to the time-sharing variables, however the region cannot be directly obtained by the time-sharing argument, and is instead obtained by a slightly more elaborate coding approach. We need the following technical definition \footnote{Here we allow a sequence of codes to achieve the $(M,R)$ pair in an asymptotic manner, {\em i.e.,} approaches this normalized memory-rate pair as the size of the file grows to infinity. Strictly speaking, this approach of definition is not necessary for us, since the quantities we obtain in (\ref{eqn:Rs}) and (\ref{eqn:Ms}) are always integers, and thus the extreme points of the constrained polytope in (\ref{eqn:LP1})-(\ref{eqn:LP5}) will always be rational, which can be achieved precisely using the proposed scheme. Then a time-sharing argument can be invoked to argue any irrational-valued memory-rate pairs in the region can be achieved. Definition \ref{def:region} however avoids this line of argument altogether.} to state the main result, which is a new inner bound to the memory-rate tradeoff region. 
\begin{definition}
\label{def:region}
A memory-rate pair $(M,R)$ is called achievable, if for any $\delta>0$, and for any sufficiently large file size $F$, there exists a code with a normalized memory size no greater than $M+\delta$ and a normalized transmission rate no greater than $R+\delta$. 
\end{definition}

\begin{theorem}
For any $t=0,1,2,\ldots,K$, any $(M,R)\in\mathcal{R}^{(t)}$ is achievable. As a consequence,  the convex closure $\mathbf{cl}\left(\cup_{t=0,\ldots,K}\mathcal{R}^{(t)}\right)$ is achievable, where $\mathbf{cl}(\cdot)$ means the convex closure. 
\label{theorem:main}
\end{theorem}

The proof of this theorem will be given in Section \ref{sec:newcode}. We also have the following corollary, whose proof is given in the appendix.
\begin{corollary}
\label{corollary:specialcase}
The memory-rate pairs in (\ref{eqn:Yu}) and those in (\ref{eqn:TianChen}) are in the region $\mathbf{cl}\left(\cup_{t=0,\ldots,K}\mathcal{R}^{(t)}\right)$.
\end{corollary}

Since $\mathcal{R}^{(t)}$ is a polytope constrained by the conditions in (\ref{eqn:LP1})-(\ref{eqn:LP5}), $\mathbf{cl}\left(\cup_{t=0,\ldots,K}\mathcal{R}^{(t)}\right)$ is also a polytope. Using standard technique \cite{Boydbook}, $\mathbf{cl}\left(\cup_{t=0,\ldots,K}\mathcal{R}^{(t)}\right)$ can be conveniently written as a region constrained by only linear constraints, and thus its boundary can be efficiently computed using linear programming.

To illustrate Theorem \ref{theorem:main}, we show that for the case $(N,K)=(3,4)$, the aforementioned new memory-rate pair $(\frac{4}{3},\frac{5}{6})$ is indeed in the region $\mathcal{R}^{(2)}$. For this purpose, we need to find a set of  $\{\alpha_{\vec{d},\vec{\mathcal{P}}^{(t)}_{\vec{d}}}\}$ such that the conditions in (\ref{eqn:LP1})-(\ref{eqn:LP5}) hold for each $\vec{d}\in\mathcal{D}$.
\begin{itemize}
\item For $\vec{d}=(A,A,B,C)=(1,1,2,3)$, let $\alpha=1$ for the decomposition pattern $\vec{\mathcal{P}}^{(2)}_{(1,1,2,3)}$
\begin{gather} 
\mathcal{P}_{(2,1,0),(1,1,2,3)}=\{\{1,2\}\}=\{\{A,B\}\},\nonumber\\
\mathcal{P}_{(2,0,1),(1,1,2,3)}=\{\{1,3\}\}=\{\{A,C\}\},\nonumber\\
\mathcal{P}_{(1,1,1),(1,1,2,3)}=\{\{1\},\{2,3\}\}=\{\{A\},\{B,C\}\},
\end{gather}
which is exactly the decomposition pattern used for (\ref{eqn:partialdecompose}). It can be verified that here 
\begin{align}
R_{\vec{d},\vec{\mathcal{P}}^{(t)}_{\vec{d}}}=1+1+[(2-1)+(1*2)]=5
\end{align}
using (\ref{eqn:Rs}), and 
\begin{align}
M_{\vec{d},\vec{\mathcal{P}}^{(t)}_{\vec{d}},1}=9-1=8,
\end{align}
where the only nonzero term comes from the transmission type $(1,1,1)$ and partition $\mathscr{P}=\{2,3\}=\{B,C\}$ in  (\ref{eqn:DeltaM}). It can be verified similarly that $M_{\vec{d},\vec{\mathcal{P}}^{(t)}_{\vec{d}},k}=8$ for $k=2,3,4$. 
\item For $\vec{d}=(A,A,B,B)=(1,1,2,2)$, two decomposition patterns are used: the first is the one without any decomposition, and the second is  
\begin{gather} 
\mathcal{P}_{(2,1,0),(1,1,2,2)}=\{\{1\},\{2\}\}=\{\{A\},\{B\}\},\nonumber\\
\mathcal{P}_{(1,2,0),(1,1,2,2)}=\{\{1\},\{2\}\}=\{\{A\},\{B\}\}.
\end{gather}
Note that this suggests a different coding approach than that used in the example of Section \ref{sec:critical}: the existence of two decomposition patterns implies that we can achieve this memory-rate pair by coding across two instances, using the two decomposition patterns given above. It is clear that for the first pattern
\begin{gather}
M_{\vec{d},\vec{\mathcal{P}}^{(t)}_{\vec{d}},k}=9,\, k=1,2,3,4,\quad \mbox{ and } R_{\vec{d},\vec{\mathcal{P}}^{(t)}_{\vec{d}}}=4,
\end{gather}
and it can be verified that for the second pattern
\begin{gather}
M_{\vec{d},\vec{\mathcal{P}}^{(t)}_{\vec{d}},k}=7,\, k=1,2,3,4,\quad \mbox{ and } R_{\vec{d},\vec{\mathcal{P}}^{(t)}_{\vec{d}}}=6.
\end{gather}
It is clear that choosing $\alpha=0.5$ for both patterns satisfies the conditions (\ref{eqn:LP1})-(\ref{eqn:LP5}).
\item For $\vec{d}=(A,A,A,C)=(1,1,1,3)$, again two decomposition patterns are used: the first is the one without any decomposition, and the other is 
\begin{gather} 
\mathcal{P}_{(2,0,1),(1,1,1,3)}=\{\{1\},\{3\}\}=\{\{A\},\{C\}\},\nonumber\\
\mathcal{P}_{(3,0,0),(1,1,1,3)}=\{\{1\}\}=\{\{A\}\}.
\end{gather}
This case is similar to the previous one, and the parameter $\alpha$ can also be chosen to be $0.5$ each.
\item For $\vec{d}=(A,A,A,A)=(1,1,1,1)$, two decomposition patterns are used: the first is the one without any decomposition, and the second is the special uncoded transmission. For the first pattern
\begin{gather}
M_{\vec{d},\vec{\mathcal{P}}^{(t)}_{\vec{d}},k}=9,\, k=1,2,3,4,\quad \mbox{ and } R_{\vec{d},\vec{\mathcal{P}}^{(t)}_{\vec{d}}}=3,
\end{gather}
and for the second pattern
\begin{gather}
M_{\vec{d},\breve{\vec{\mathcal{P}}}^{(t)}_{\vec{d}},k}=3,\, k=1,2,3,4,\quad \mbox{ and } R_{\vec{d},\vec{\vec{\mathcal{P}}}^{(t)}_{\vec{d}}}=12.
\end{gather}
We can choose $\alpha=\frac{5}{6}$ for the first pattern and the conditions (\ref{eqn:LP1})-(\ref{eqn:LP5}) indeed hold.
\end{itemize}

\section{The New Coding Scheme}
\label{sec:newcode}

We first give the prefetching strategy and the delivery strategy. The correctness of the code is then proved, which establishes Theorem \ref{theorem:main}.
\subsection{The Prefetching Strategy}

The prefetching strategy is in fact rather straightforward, which is to encode the symbols allocated to a user using a rank metric code to produce the linear combinations. However, since we allow coding across multiple instances, a technical issue arises as what are the proportions of different delivery patterns. These values are needed to determine two parameters: the total number of instances to code across, and the total number of coded symbols to cache. To address this technical issue, we consider the following line of argument.

Suppose a memory-rate tradeoff pair $(M,R)\in\mathcal{R}^{(t)}$. The definition of $\mathcal{R}^{(t)}$ implies that there exists a set of $\{\alpha_{\vec{d},\vec{\mathcal{P}}^{(t)}_{\vec{d}}}\}$ for which the conditions in (\ref{eqn:LP1})-(\ref{eqn:LP5}) hold. Let us assume that a positive integer $r$ is chosen such that there exists a set of non-negative integers $\{r_{\vec{d},\vec{\mathcal{P}}^{(t)}_{\vec{d}}}\}$ 
\begin{align}
\left|\frac{r_{\vec{d},\vec{\mathcal{P}}^{(t)}_{\vec{d}}}}{r}-\alpha_{\vec{d},\vec{\mathcal{P}}^{(t)}_{\vec{d}}}\right|\leq \epsilon.
\end{align}
Clearly $\epsilon$ can be arbitrarily small by choosing $r$ sufficiently large. Essentially, during the delivery phase, for each demand type $\vec{d}$, within the total of $r$ instances that are being coded across, we will use the decomposition pattern $\vec{\mathcal{P}}^{(t)}_{\vec{d}}$ on $r_{\vec{d},\vec{\mathcal{P}}^{(t)}_{\vec{d}}}$ of them during delivery. 

Let us now fix $r$ and $\{r_{\vec{d},\vec{\mathcal{P}}^{(t)}_{\vec{d}}}\}$. For $\vec{d}\in{\mathcal{D}}$, define the memory-rate pair
\begin{align}
(M'_{\vec{d}},R'_{\vec{d}})\triangleq\frac{1}{r{K \choose t}}\left(\max_{k\in I_K}\sum_{\vec{\mathcal{P}}^{(t)}_{\vec{d}}} r_{\vec{d},\vec{\mathcal{P}}^{(t)}_{\vec{d}}} M_{\vec{d},\vec{\mathcal{P}}^{(t)}_{\vec{d}},k},\sum_{\vec{\mathcal{P}}^{(t)}_{\vec{d}}} r_{\vec{d},\vec{\mathcal{P}}^{(t)}_{\vec{d}}} R_{\vec{d},\vec{\mathcal{P}}^{(t)}_{\vec{d}}}\right),
\end{align}
Let us also define $M'_r\triangleq \max_{\vec{d}\in {\mathcal{D}}}M'_{\vec{d}}$ and $R'_r\triangleq \max_{\vec{d}\in{\mathcal{D}}}R'_{\vec{d}}$, which will be the effective memory-rate pair of this code. 

The key design constraint is that the prefetching strategy needs to be independent of the demand vector, which we describe next. In the proposed code, each file contains $r{K \choose t}$ symbols.  Each symbol is thus denoted as $W^{(i)}_{n,\mathcal{S}}$, where $i\in I_r$, $n\in I_N$ and $\mathcal{S}\subseteq I_K$ with $|\mathcal{S}|=t$, is assumed to be a symbol in $\mathbb{F}_{2^m}$ for some sufficiently large $m$ to be specified shortly. Each file symbol (segment) will be provided to $t$ users as indicated by $\mathcal{S}$, to be stored as a component of some linear combinations. There are a total of 
\begin{align*}
P\triangleq rN{K-1\choose t-1}
\end{align*}
symbols allocated to each user, however, only $P_o-P$ linear combinations of them are stored in the cache, and the parameter $P_o$ is directly related to the  normalized memory $M'_r$ as
\begin{align}
P_o-P=r M'_r{K \choose t}.
\end{align}
Note that $P_o$ is always an integer.  
A $(P_o,P)$ systematic rank metric code is then used to encode the $P$ symbols at each user, and the $P_o-P$ parities of this code are placed at each user's cache. For such a rank metric code to exist, $m\geq P_o$ suffices. 

Our plan next is to show that for each $\vec{d}\in {\mathcal{D}}$, a valid delivery strategy exists with a delivery rate $R'_{\vec{d}}$. Then by making the integer $r$ sufficiently large, and choosing the integers $\{r_{\vec{d},\vec{\mathcal{P}}^{(t)}_{\vec{d}}}\}$ appropriately such that $\epsilon\geq 0$ is made arbitrarily small, we have
\begin{align}
\lim_{r\rightarrow\infty} (M'_r,R'_r)=\frac{1}{{K \choose t}}\left(\max_{\vec{d}\in{\mathcal{D}}}\max_{k\in I_K}\sum_{\vec{\mathcal{P}}^{(t)}_{\vec{d}}} \alpha_{\vec{d},\vec{\mathcal{P}}^{(t)}_{\vec{d}}} M_{\vec{d},\vec{\mathcal{P}}^{(t)}_{\vec{d}},k},\max_{\vec{d}\in{\mathcal{D}}}\sum_{\vec{\mathcal{P}}^{(t)}_{\vec{d}}} \alpha_{\vec{d},\vec{\mathcal{P}}^{(t)}_{\vec{d}}} R_{\vec{d},\vec{\mathcal{P}}^{(t)}_{\vec{d}}}\right) \preceq (M,R).
\end{align}
This would prove that the targeted $(M,R)$ is indeed achievable. 

\subsection{The Delivery Strategy}
\label{sec:delivery}

Consider any demand vector $\vec{d}\in \mathcal{D}$, and recall the parameters $\{r_{\vec{d},\vec{\mathcal{P}}^{(t)}_{\vec{d}}}\}$ have been chosen. For convenience, suppose there are a total of $q$ possible decomposition patterns for the demand vector $\vec{d}$, with the first one as 
$\vec{\mathcal{P}}^{(t)}_{\vec{d},1}=\breve{\vec{\mathcal{P}}}^{(t)}_{\vec{d}}$ which is the special case associated with the uncoded transmissions, and $\vec{\mathcal{P}}^{(t)}_{\vec{d},j}$, for $j=2,3,\ldots,q$ other decomposition patterns. For a specific transmission type $\vec{t}$,  the corresponding decomposition in $\vec{\mathcal{P}}^{(t)}_{\vec{d},j}$ is written as $\mathcal{P}_{\vec{t},\vec{d},j}$. Note that $\sum_{j=1}^q r_{\vec{d},\vec{\mathcal{P}}^{(t)}_{\vec{d},j}}=r$. 
For the demand vector $\vec{d}$, the transmissions in the proposed scheme are as given in Algorithm \ref{alg:trans}.
\begin{figure}
\begin{algorithm}[H]
\caption{The delivery strategy}
\label{alg:trans}
\KwIn{$t$, $\vec{d}$, $\{r_{\vec{d},\vec{\mathcal{P}}^{(t)}_{\vec{d}}}\}$, and $\{W^{(i)}_{n,\mathcal{S}}\}$}
Compute $\mathcal{\vec{m}}$, $\mathbb{supp}(\mathcal{\vec{m}})$, and $N^*$ from $\vec{d}$.\\
\eIf{$K-t\leq N^*-1$}{
	\For{$\mathcal{S}\subseteq I_K$: $|\mathcal{S}|=t$}{
		Find a set $\mathcal{A}^*\subseteq \mathbb{supp}(\mathcal{\vec{m}})$ such that $\cup_{n\in\mathcal{A}^*}I^{[n]}\subseteq \mathcal{S}$ and $|\mathcal{A}^*|=N^*-K+t$\\
		\For{$i=1$ to $r_{\vec{d},\breve{\vec{\mathcal{P}}}^{(t)}_{\vec{d}}}$}{	
			Transmit $W^{(i)}_{n,\mathcal{S}}$, all $n\in\mathbb{supp}(\mathcal{\vec{m}})\setminus \mathcal{A}^*$.
		}
	}
}{
	Choose a set $\mathcal{A}\subseteq I_N$, such that $|\mathcal{A}|=\min(K-t,\tilde{N})$ and $\mathbb{supp}(\mathcal{\vec{m}})\subseteq \mathcal{A}$\\
	\For{$n\in\mathcal{A}$}{	
		\For{$i=1$ to $r_{\vec{d},\breve{\vec{\mathcal{P}}}^{(t)}_{\vec{d}}}$}{		
			Transmit $W^{(i)}_{n,\mathcal{S}}$, all $\mathcal{S}\subseteq I_K$ such that $\mathcal{S}=t$.
		}
	}
}
\For{$j=2$ to $q$}{
	\For{$\vec{t}\in \mathcal{T}_{\vec{d}}^{(t)}$}{
		\For{$\mathscr{P} \in \mathcal{P}_{\vec{t},\vec{d},j}$}{    
			\For{ $\mathcal{B}: \mathcal{T}(\mathcal{B})=\vec{t}$, $\left(\bigcup_{n\in \mathscr{P}} \{\ell^{[n]}\}\right)\cap \mathcal{B}\neq \emptyset$}{
				\For{$i=\sum_{k=1}^{j-1}r_{\vec{d},\vec{\mathcal{P}}^{(t)}_{\vec{d}},k} +1$ to $\sum_{k=1}^{j}r_{\vec{d},\vec{\mathcal{P}}^{(t)}_{\vec{d}},k}$}{				
					Transmit $\oplus_{n\in\mathscr{P}}\left(\oplus_{k\in \mathcal{B}\cap I^{[n]}} W^{(i)}_{n,\mathcal{B}\setminus k}\right)$;
			       }
			}      
		}
	}
}
\end{algorithm}
\end{figure}

The transmissions on line-6 and line-13 in Algorithm \ref{alg:trans} are uncoded, which stem from the special transmission pattern $\breve{\vec{\mathcal{P}}}^{(t)}_{\vec{d}}$. We first need to show that the steps on line-4 and line-10 are valid,  {\em i.e.,} such a set $\mathcal{A}^*$ or $\mathcal{A}$ can always be found. The latter case is immediate by observing that in this case $K-t\geq N^*$, and thus $N^*\leq \min(K-t,\tilde{N})\leq \tilde{N}$, and we can always find a set $\mathcal{A}$ such that $\mathbb{supp}(\vec{m})\subseteq\mathcal{A}\subseteq I_N$. To see that the step on line-4 is also valid, first observe that in this case $K-t\leq N^*-1$, and we need to find a set of files $\mathcal{A}^*$, such that the given set of users $\mathcal{S}$ (where $|\mathcal{S}|=t$) includes all the users that request files in $\mathcal{A}^*$. Suppose we cannot find such a set, this means that there are less than $N^*-K+t$ such files, or more than $N^*-(N^*-K+t)=K-t$ files that are being requests by some users not in $\mathcal{S}$, but this is impossible, since there are only $K-t$ users not in the set $\mathcal{S}$. Thus the supposition is not true, and we can always find such a set $\mathcal{A}^*$. It is straightforward to count the total number of transmissions as $r_{\vec{d},\breve{\vec{\mathcal{P}}}^{(t)}_{\vec{d}}}\min(K-r,\tilde{N}){K \choose t}$, when Algorithm \ref{alg:trans} completes line-16. 

The transmissions on line-22 in Algorithm \ref{alg:trans} have the following property: at least one of the component $W_{n,\mathcal{B}\setminus k}$ (for any fixed superscript $^{(i)}$) in the transmission must have $k$ that is a leader. It is a simple combinatorial counting task to show that the total number of transmissions in this part of the algorithm is given by $\sum_{j=2}^q r_{\vec{d},\vec{\mathcal{P}}^{(t)}_{\vec{d},j}}R_{\vec{d},{\vec{\mathcal{P}}}^{(t)}_{\vec{d},j}}$. Essentially we examine all the transmission types, for which the decomposition pattern for the transmission type follows $\vec{\mathcal{P}}^{(t)}_{\vec{d},j}$ in the corresponding instances (indexed by the superscript $^{(i)}$), then count the transmitted linear combinations associated with each partition in this decomposition. We must eliminate the transmissions where no leader in this partition $\mathscr{P}$ is included, which is indeed accounted for as the term $-\prod_{n \in \mathscr{P}} \binom{m_n-1}{t_n}$ in $R_{\vec{d},{\vec{\mathcal{P}}}^{(t)}_{\vec{d},j}}$. Thus after the algorithm runs to completion, a total of $\sum_{j=1}^q r_{\vec{d},\vec{\mathcal{P}}^{(t)}_{\vec{d},j}}R_{\vec{d},{\vec{\mathcal{P}}}^{(t)}_{\vec{d},j}}$ symbols are transmitted. 

The transmissions on line-22 in Algorithm \ref{alg:trans} are a subset of the decomposed transmissions given in (\ref{eqn:directdecomp}) for the patterns $\vec{\mathcal{P}}^{(t)}_{\vec{d},j}$,  since transmissions without any leader are not allowed as mentioned early. This removes the redundancy in the transmissions after a native decomposition. The precise linear independence relations can be captured in a set of lemmas given in the sequel.

\subsection{Three Auxiliary Lemmas} 

When stating these lemmas, we omit the superscript $^{(i)}$ which is used to index the code instances that are being coded across, as well as the decomposition pattern index $j=2,3,\ldots,q$, since they are irrelevant in these settings. We shall return to this notation later on when it becomes important. 

\begin{lemma}[Redundancy Reduction Lemma]
Fix a demand vector $\vec{d}$ and a valid transmission type $\vec{t}$. Designate a subset $\mathscr{P}\subseteq\mathbb{supp}(\vec{t})$ as the variable set, and $\bar{\mathscr{P}}=\mathbb{supp}(\vec{t})\setminus \mathscr{P}$ as the fixed set. Further fix an arbitrary subset $\mathcal{A}\subseteq \cup_{n\in \bar{\mathscr{P}}} I^{[n]}$ such that $\mathcal{A}\cap I^{[n]}=t_n$ for all $n\in \bar{\mathscr{P}}$. Let $\mathcal{L}\triangleq \cup_{n\in \mathscr{P}}\{\ell^{[n]}\}$ be the leader set. Let $\mathcal{Q}_n\subseteq I^{[n]}\setminus \ell^{[n]}$ be any subset such that $|\mathcal{Q}|=t_n$, and let $\mathcal{Q}\triangleq \cup_{n\in\mathscr{P}} \mathcal{Q}_n$. 
The following equation holds
\begin{align}
\oplus_{\substack{\mathcal{V}\subseteq \mathcal{Q}\cup\mathcal{L}: \\|\mathcal{V}\cap I^{[n]}|=t_n,\\
\forall n\in \mathscr{P}}}\oplus_{k\in\mathcal{V}}W_{d_k,\mathcal{V}\cup \mathcal{A}\setminus\{k\}}=0.\label{eqn:rrl}
\end{align}
\label{lemma:redundancyreduction}
\end{lemma}
\begin{proof}
Observe that
\begin{align}
LHS&=
\oplus_{\ell\in \mathcal{L}}\oplus_{\substack{\mathcal{V}\subseteq \mathcal{Q}\cup\mathcal{L}: \\|\mathcal{V}\cap I^{[n]}|=t_n,\\
\forall n\in \mathscr{P}}}\oplus_{k\in\mathcal{V}\cap I^{[\ell]}}W_{d_{\ell},\mathcal{V}\cup \mathcal{A}\setminus\{k\}}
\end{align}
Consider any fixed $\ell\in \mathcal{L}$, and enumerate all set $\mathcal{V}$ by parts $\mathcal{V}\triangleq(\hat{\mathcal{V}},\tilde{\mathcal{V}})$
\begin{align}
&\oplus_{\substack{\mathcal{V}\subseteq \mathcal{Q}\cup\mathcal{L}: \\|\mathcal{V}\cap I^{[n]}|=t_n,\\
\forall n\in \mathscr{P}}}\oplus_{k \in \mathcal{V}\cap I^{[\ell]}}W_{d_\ell,\mathcal{V}\cup \mathcal{A}\setminus\{k\}}\nonumber\\
&\quad=\oplus_{\substack{\hat{\mathcal{V}}\subseteq (\mathcal{Q}\setminus \mathcal{Q}_{d_\ell})\cup(\mathcal{L}\setminus\{\ell\}):\\
|\mathcal{V}\cap I^{[n]}|=t_n,\\
\forall n\in \mathscr{P}\setminus\{d_\ell\}}}\left(\oplus_{\substack{\tilde{\mathcal{V}}\subseteq \mathcal{Q}_{d_\ell}\cup \{\ell\}: \\|\tilde{\mathcal{V}}|=t_{d_{\ell}} }}\left(\oplus_{k\in \tilde{\mathcal{V}}}W_{d_\ell,(\hat{\mathcal{V}}\cup \mathcal{A})\cup(\tilde{\mathcal{V}}\setminus\{k\})}\right)\right).
\end{align}
Now consider a fixed $\hat{\mathcal{V}}$, and consider the inner summation
\begin{align}
\label{eqn:summationsymmetric}
\oplus_{\substack{\tilde{\mathcal{V}}\subseteq \mathcal{Q}_{d_\ell}\cup \{\ell\}: \\|\tilde{\mathcal{V}}|=t_{d_{\ell}} }}\left(\oplus_{k\in \tilde{\mathcal{V}}}W_{d_\ell,(\hat{\mathcal{V}}\cup \mathcal{A})\cup(\tilde{\mathcal{V}}\setminus\{k\})}\right).
\end{align}
This is a summation of $(t_{d_\ell}+1)t_{d_\ell}$ file symbols, each of which is in the form $W_{d_\ell,(\hat{\mathcal{V}}\cup \mathcal{A})\cup \dot{\mathcal{V}}}$, where $\dot{\mathcal{V}}$ is a subset of $\mathcal{Q}_{d_\ell}\cup \{\ell\}$ such that $|\dot{\mathcal{V}}|=t_{d_\ell}-1$. Since $|\mathcal{Q}_{d_\ell}\cup \{\ell\}|=t_{d_\ell}+1$ and the summation form in (\ref{eqn:summationsymmetric}) is symmetric, each file symbol appears exactly twice, which cancel out each other in this binary (extension) field. The proof is thus complete.
\end{proof}

The above lemma can be used to show that the decomposed transmissions without any leaders are redundant. To see this, notice that (\ref{eqn:rrl}) can be rewritten as
\begin{align}
\left(\oplus_{\substack{\mathcal{V}\subseteq \mathcal{Q}\cup\mathcal{L}: \mathcal{V}\neq \mathcal{Q}
\\|\mathcal{V}\cap I^{[n]}|=t_n, \forall n\in \mathscr{P}}}\oplus_{k\in\mathcal{V}}W_{d_k,\mathcal{V}\cup \mathcal{A}\setminus\{k\}}\right)\oplus\left(\oplus_{k\in\mathcal{Q}}W_{d_k,\mathcal{Q}\cup \mathcal{A}\setminus\{k\}}\right)=
0
\end{align}
Clearly the summation in the second bracket, which is one of decomposed parts from (\ref{eqn:directdecomp}) without any leaders, can be expressed as a linear combination of those in the first bracket, which all have some leaders and are indeed in the delivery transmissions given on line-22 in Algorithm \ref{alg:trans}. Conversely, the transmissions obtained by directly decomposing those in the delivery transmissions of \cite{MaddahAliNiesen:14} can be reconstructed using the transmissions given on line-22 in Algorithm \ref{alg:trans}. 
Lemma \ref{lemma:redundancyreduction} is a generalized version of a similar lemma in \cite{Yu:16}, which was used to remove the redundancy in the coding scheme given in \cite{MaddahAliNiesen:14}. 

The next two lemmas essentially state that there is no further linear redundancy in the transmissions in line-22 of Algorithm \ref{alg:trans} to be removed. In order to state the lemmas, the following definition is needed. For any fixed $\vec{d}$, $\vec{t}$, $\mathscr{P}\in \mathcal{P}_{\vec{t},\vec{d}}$, and $\mathcal{A}\subseteq \cup_{n\in \mathbb{supp}(\vec{t})\setminus \mathscr{P}}I^{[n]}$ for which $\mathcal{A}\cap I^{[n]}=t_n$ for all $n\in \mathbb{supp}(\vec{t})\setminus \mathscr{P}$, let
\begin{align}
\mathcal{W}_{\vec{d},\vec{t},\mathscr{P},\mathcal{A}}\triangleq \bigcup_{\substack{\mathcal{B}\subseteq \cup_{n\in \mathscr{P}}I^{[n]}:\\
\mathcal{B}\cap I^{[n]}=t_n,\\\forall n\in \mathscr{P}}}\left\{W_{d_k,\mathcal{A}\cup\mathcal{B}\setminus k}: d_k\in\mathscr{P}\right\}. \label{eqn:Wset}
\end{align}
The next lemma states that the decomposed transmissions can in fact be separated naturally into mutually exclusive groups.

\begin{lemma}
For any $\vec{d}$ and $\vec{\mathcal{P}}^{(t)}_{\vec{d}}$, and any $(\vec{t'},\mathscr{P}',\mathcal{A}')\neq (\vec{t''},\mathscr{P}'',\mathcal{A}'')$, where $\mathscr{P}'\in \mathcal{P}_{\vec{t'},\vec{d}}$ and $\mathscr{P}''\in \mathcal{P}_{\vec{t''},\vec{d}}$, we have 
$\mathcal{W}_{\vec{d},\vec{t'},\mathscr{P}',\mathcal{A}'}\cap \mathcal{W}_{\vec{d},\vec{t''},\mathscr{P}'',\mathcal{A}''}=\emptyset$.
\label{lemma:separation}
\end{lemma}
\begin{proof}

Suppose that the two sets have a common element $W_{n,\mathcal{C}}$. Then this implies that, 
\begin{align}
&t'_n=t''_n=|\{\hat{k}\in\mathcal{C}: d_{\hat{k}}=n\}|+1,\nonumber\\
&t'_i=t''_i=|\{\hat{k}\in\mathcal{C}: d_{\hat{k}}=i\}|,\quad i\in I_N,\quad i\neq n.
\end{align}
{\em i.e., } $\vec{t'}=\vec{t''}$; let us write this transmission type as $\vec{t}$. It also follows that $\mathscr{P}'=\mathscr{P}''$, since $n\in\mathscr{P}'$ and $n\in\mathscr{P}''$, but $\mathscr{P}'\cap \mathscr{P}''=\emptyset$ if $\mathscr{P}'$ and $\mathscr{P}''$ are distinct; we can thus denote this partition as $\mathscr{P}$. It further follows that $\mathcal{A}'=\mathcal{A}''=\mathcal{C}\cap \cup_{n\in \mathbb{supp}(\vec{t})\setminus \mathscr{P}}I^{[n]}$. This is a contradiction, and thus there is no common element between the two sets. The proof is thus complete.
\end{proof}

\begin{lemma}\label{lemma:independenceLong}
Each transmission on line-22 in Algorithm \ref{alg:trans} is a linear combination of the elements in a single set $\mathcal{W}_{\vec{d},\vec{t},\mathscr{P},\mathcal{A}}$. All the linear combinations in the transmissions on line-22 of Algorithm \ref{alg:trans} using symbols in a single set $\mathcal{W}_{\vec{d},\vec{t},\mathscr{P},\mathcal{A}}$ are linearly independent. 
\end{lemma}

\begin{proof}
The first statement is through direct inspection. We can prove the second statement by analyzing the rank of the corresponding coding matrix, which is however rather long and tedious. We instead prove it through a shortcut, directly utilizing the optimality result established in \cite{Yu:16}. 

Fix a demand vector $\vec{d}\in{\mathcal{D}}$, a transmission type $\vec{t}$, and a partition $\mathscr{P}\in \mathcal{P}_{\vec{t},\vec{d}}$. We only need to prove that for a fixed $\mathcal{A}$, the transmissions 
\begin{align}
\oplus_{n\in\mathscr{P}}\left(\oplus_{k\in \mathcal{B}\cap I^{[n]}} W_{n,(\mathcal{B}\setminus k) \cup \mathcal{A}}\right), \label{eqn:removeA}
\end{align}
when $\mathcal{B}$ ranges over all subsets of $\cup_{n\in\mathscr{P}}I^{[n]}$ that satisfy the condition
\begin{align}
\mathcal{B}\subseteq \cup_{n\in \mathscr{P}}I^{[n]}: \mathcal{B}\cap I^{[n]}=t_n, \forall n\in \mathscr{P}\label{eqn:rangeB}
\end{align}
are indeed linearly independent. For this purpose, the exact choice of $\mathcal{A}$ is not relevant, and thus we might as well simply drop it by defining
\begin{align}
\hat{W}_{n,\mathcal{B}\setminus k }\triangleq W_{n,(\mathcal{B}\setminus k) \cup \mathcal{A}},
\end{align}
which lead to the representation
\begin{align}
\oplus_{n\in\mathscr{P}}\left(\oplus_{k\in \mathcal{B}\cap I^{[n]}} \hat{W}_{n,\mathcal{B}\setminus k }\right), \label{eqn:Aremoved}
\end{align}
where $\mathcal{B}$ has the same range as (\ref{eqn:rangeB}). Now consider a caching system that has files $\{\hat{W}_n:n\in \mathscr{P}\}$, the users $\cup_{n\in \mathscr{P}} I^{[n]}$, and the demand vector formed by taking the demand vector $\vec{d}$ at the coordinates $\cup_{n\in \mathscr{P}} I^{[n]}$. The transmissions (\ref{eqn:Aremoved}) are in fact part of the transmissions in the scheme in \cite{Yu:16} for this system when choosing $t=|\mathscr{P}|-1$. These transmissions cannot possibly be linearly dependent, because if so, the dependence could have been removed to further improve the delivery transmission rate, but it was shown in \cite{Yu:16} that this transmission scheme is in fact optimal for each demand vector. The proof is thus complete.
\end{proof}

\subsection{The Correctness of the Coding Scheme}

The next proposition shows that the code is indeed valid for any $\vec{d}\in{\mathcal{D}}$. 
\begin{prop}
Each user can use the delivery transmissions in Algorithm \ref{alg:trans} and the cached content to recover the requested file for any $\vec{d}\in{\mathcal{D}}$. 
\end{prop}
\begin{proof}
Consider an arbitrary user $k_o$, whose demands is $d_{k_o}$. From the transmissions in line-1 to line-16, the user can clearly collect for each $\mathcal{S}$, where $k_o\in \mathcal{S}$, a total of $r_{\vec{d},\breve{\vec{\mathcal{P}}}^{(t)}_{\vec{d}}} \min(K-r,\tilde{N})$ uncoded symbols, in the form of $W^{i}_{n,\mathcal{S}}$ for $n\in\mathbb{supp}(\vec{m})\setminus\mathcal{A}^*$ (or $n\in\mathcal{A}$), and there are clearly ${K-1 \choose t-1}$ possible ways to choose such a $\mathcal{S}$. Thus a total of $r_{\vec{d},\breve{\vec{\mathcal{P}}}^{(t)}_{\vec{d}}} \Delta M_{\vec{d},\breve{\vec{\mathcal{P}}}^{(t)}_{\vec{d}}, k}$ symbols are collected. 

Then consider the transmissions on line-22 in the algorithm. For each transmission type $\vec{t}$ such that $d_{k_o}\in \mathbb{supp}(\vec{t})$, consider every $\mathcal{B}$ such that $\mathcal{T}(\mathcal{B})=\vec{t}$, $|\mathcal{B}|=t+1$, $k_o\in \mathcal{B}$, and $\left(\bigcup_{n\in \mathscr{P}} \{\ell^{[n]}\}\right)\cap \mathcal{B}\neq \emptyset$. User-$k$ collects the following transmissions in the corresponding instances:
\begin{align}
\oplus_{n\in\mathscr{P}}\left(\oplus_{k\in \mathcal{B}\cap I^{[n]}} W^{(i)}_{n,\mathcal{B}\setminus k}\right), \mathscr{P} \mbox{ such that } \mathscr{P}\in {\mathcal{P}}_{\vec{t},\vec{d},j} \mbox{ and } d_{k_o}\notin \mathscr{P};\label{eqn:part1}
\end{align}
First note that since $d_{k_o}\notin \mathscr{P}$ but $n\in \mathscr{P}$, and $k\in I^{[n]}$, we have $k_o\neq k$ in the inner enumeration. Thus $k_o\in \mathcal{B}\setminus k$. This implies that all the collected transmissions are linear combinations of the  symbols of the form $W^{(i)}_{n,\mathcal{S}}$ where $k_o\in \mathcal{S}$, which are the components of the linear combinations stored in the cache of user-$k_o$. 
It is straightforward to count that user-$k$ collects a total of $\Delta M_{\vec{d},{\vec{\mathcal{P}}}^{(t)}_{\vec{d},j}, k_o}$ for the decomposition pattern ${\vec{\mathcal{P}}}^{(t)}_{\vec{d},j}$
such transmissions in (\ref{eqn:part1}) for each fixed $i$ value. Together with the cached contents, user-$k_o$ has a total of 
\begin{align}
M'+\sum_{j=1}^q r_{\vec{d},{\vec{\mathcal{P}}}^{(t)}_{\vec{d},j}}\Delta M_{\vec{d},{\vec{\mathcal{P}}}^{(t)}_{\vec{d},j}, k_o}\geq rN{K-1 \choose t-1},
\end{align}
linear combinations of the $rN{K-1 \choose t-1}$ symbols $W^{(i)}_{n,\mathcal{S}}$ where $k_o\in \mathcal{S}$. The collected linear combinations from the delivery transmissions are clearly linearly independent due to Lemma \ref{lemma:separation} and Lemma \ref{lemma:independenceLong}. Thus these linear combinations that user-$k_o$  has gathered can be viewed as a full rank transformation of the corresponding rank metric code symbols, which were produced at the prefetching stage by user-$k_o$. By Lemma \ref{lemma:fullrank}, user $k_o$ can recover all these symbols to their uncoded form. At this point, user-$k_o$ essentially has all the symbols as if the uncoded prefetching strategy in \cite{MaddahAliNiesen:14} was used. 
It remains to argue that user-$k_o$ can also recover the file symbols of file  $d_{k_o}$, which is in the form $W^{(i)}_{d_{k_o},\mathcal{S}}$ such that $k_o\notin \mathcal{S}$. This is straightforward to check for $i=1,\ldots,r_{\vec{d},\breve{\vec{\mathcal{P}}}^{(t)}_{\vec{d}}}$, because all the needed symbols are transmitted in the uncoded form. For the other cases, since by Lemma \ref{lemma:redundancyreduction}, the original transmissions in the delivery scheme in \cite{MaddahAliNiesen:14} can be completely reconstructed using the transmissions in Algorithm \ref{alg:trans}, and with these transmissions and the uncoded prefetched symbols in the cache, user-$k_o$ can indeed recover the missing file symbols through the decoding strategy in \cite{MaddahAliNiesen:14}. The proof is thus complete.
\end{proof}

\section{Conclusion}
\label{sec:conclusion}
We discovered a connection between the caching strategy in \cite{TianChen:16Arxiv} and that in \cite{Yu:16}, that is decomposing the delivery transmissions in \cite{Yu:16} yields those in \cite{TianChen:16Arxiv} in some cases. This allows us to view the coding strategy in \cite{TianChen:16Arxiv} and that in \cite{Yu:16} as the two extremes of a more general scheme. The general scheme can achieve some memory-rate pairs previously unknown in the literature, and can be computed using a linear programming approach. 

We note that although the new scheme unifies the codes in \cite{TianChen:16Arxiv} and \cite{Yu:16}, it does not appear to include the codes in \cite{gomez2016fundamental}. We suspect that an improved code can be found by  analyzing the transmission types more carefully to optimize explicitly the decomposition patterns, and then incorporate certain coding components in \cite{gomez2016fundamental}; this is part of our ongoing work. Our work reported here is information theoretic in nature, and little attention has been paid to the complexity of the codes. Particularly, the proposed code has a large alphabet size, a large subpacketization factor, and needs to code across a large number of instances. Such a code is challenging to use directly in practical systems, and effort toward simplifying it appears worthwhile. 

\section*{Appendix: Proof of Corollary \ref{corollary:specialcase}}
We first show that (\ref{eqn:Yu}), {\em i.e.}, the memory-rate tradeoff points given in \cite{Yu:16}, can be obtained by specializing (\ref{eqn:Rs}), (\ref{eqn:Ms}), and (\ref{eqn:LP1})-(\ref{eqn:LP5}). For this case, the decomposition patterns are given by $\mathcal{P}_{\vec{t},\vec{d}}=\{\mathbb{supp}(\vec{t})\}$, {\em i.e.,} there is no decomposition. As such, (\ref{eqn:Rs}) reduces to
\begin{equation}
R_{\vec{d},\vec{\mathcal{P}}^{(t)}_{\vec{d}}} = \sum_{\vec{t} \in \mathcal{T}^{(t)}_{\vec{d}}}  \prod_{n \in \mathbb{supp}(\vec{t})} \binom{m_n}{t_n} - \sum_{\vec{t} \in \mathcal{T}^{(t)}_{\vec{d}}}
\prod_{n \in \mathbb{supp}(\vec{t})} \binom{m_n-1}{t_n}.\label{eqn:specialYu}
\end{equation}
The first summation is clearly ${K \choose t+1}$, since it is simply the number of ways to choose $(t+1)$ users from the $K$ users, however counted one transmission type at a time. To simplify the second term in (\ref{eqn:specialYu}), let us consider any demand vector $\vec{d}\in\mathcal{D}$. For any transmission type $\vec{t}$ where there exists $n\in \mathbb{supp}(\vec{t})$ such that $t_n=m_n$, the product $\prod_{n \in \mathbb{supp}(\vec{t})} \binom{m_n-1}{t_n}$ is clearly zero. The second summation can thus be viewed as counting the number of ways to choose $(t+1)$ users, however, with the leaders $\{\ell^{[n]},m_n\neq 0\}$ not being chosen; there is clearly ${K-N^* \choose t+1}$ ways to do so. This implies that
\begin{align}
R_{\vec{d},\vec{\mathcal{P}}^{(t)}_{\vec{d}}} = {K \choose t+1}-{K-N^* \choose t+1}\leq {K \choose t+1}-{K-\tilde{N} \choose t+1}, \quad \forall \vec{d}\in \mathcal{D}.\label{eqn:Ryu}
\end{align}
Similarly (\ref{eqn:Ms}) can be simplified for any $\vec{d}\in\mathcal{D}$ as
\begin{align}
M_{\vec{d}, \vec{\mathcal{P}}^{(t)}_{\vec{d}}, k} = N \binom{K-1}{t-1},\label{eqn:Myu}
\end{align}
since the other term disappears with the choice $\mathcal{P}_{\vec{t},\vec{d}}=\{\mathbb{supp}(\vec{t})\}$. 
The quantities (\ref{eqn:Ryu}) and (\ref{eqn:Myu}) are independent of $\vec{d}$. It is clear that (\ref{eqn:Yu}) is identical to (\ref{eqn:Ryu}) and (\ref{eqn:Myu}) after normalization with the file size $F={K\choose t}$. It is easy to verify that they are indeed in the region $\mathcal{R}^{(t)}$, thus it is clearly inside $\mathbf{cl}\left(\cup_{t=0,\ldots,K}\mathcal{R}^{(t)}\right)$. 

Next we show that (\ref{eqn:TianChen}),  {\em i.e.}, the memory-rate tradeoff points given in \cite{TianChen:16Arxiv}, can also be obtained by specializing (\ref{eqn:Rs}), (\ref{eqn:Ms}), and (\ref{eqn:LP1})-(\ref{eqn:LP5}). In this case, the decomposition patterns are given by $\mathcal{P}_{\vec{t},\vec{d}}=\{\{n\}:n\in\mathbb{supp}(\vec{t})\}$, {\em i.e.,} $\mathbb{supp}(\vec{t})$ is partitioned into sets, where each set is a singleton. It follows that
\begin{align}
R_{\vec{d},\vec{\mathcal{P}}^{(t)}_{\vec{d}}}&=\sum_{\vec{t} \in \mathcal{T}^{(t)}_{\vec{d}}} \sum_{n\in \mathbb{supp}(\vec{t})} \left[\Bigg( \binom{m_n}{t_n} - \binom{m_n-1}{t_n} \Bigg)\cdot\prod_{n'\in \mathbb{supp}(\vec{t})\setminus \{n\}}{m_{n'} \choose t_{n'}}\right]\nonumber\\
&=\sum_{\vec{t} \in \mathcal{T}^{(t)}_{\vec{d}}} \sum_{n\in \mathbb{supp}(\vec{t})} \left[\binom{m_n-1}{t_n-1}\cdot\prod_{n'\in \mathbb{supp}(\vec{t})\setminus \{n\}}{m_{n'} \choose t_{n'}}\right].
\end{align}
Define $\mathbb{1}_{c}$ to be the indicator function which is equal to $1$ when the condition $c$ holds, and is equal to $0$ otherwise. We can now rewrite the summation as
\begin{align}
R_{\vec{d},\vec{\mathcal{P}}^{(t)}_{\vec{d}}}&=\sum_{\vec{t} \in \mathcal{T}^{(t)}_{\vec{d}}} \sum_{n\in\mathbb{supp}(\vec{m})} \mathbb{1}_{n\in \mathbb{supp}(\vec{t})}\left[\binom{m_n-1}{t_n-1}\cdot\prod_{n'\in \mathbb{supp}(\vec{t})\setminus \{n\}}{m_{n'} \choose t_{n'}}\right]\nonumber\\
&=\sum_{n\in\mathbb{supp}(\vec{m})}\sum_{\vec{t} \in   \mathcal{T}^{(t)}_{\vec{d}}} \mathbb{1}_{n\in \mathbb{supp}(\vec{t})}\left[\binom{m_n-1}{t_n-1}\cdot\prod_{n'\in \mathbb{supp}(\vec{t})\setminus \{n\}}{m_{n'} \choose t_{n'}}\right]. 
\end{align}
Notice that the equality
\begin{align}
\sum_{\vec{t} \in   \mathcal{T}^{(t)}_{\vec{d}}} \mathbb{1}_{n\in \mathbb{supp}(\vec{t})}\left[\binom{m_n-1}{t_n-1}\cdot\prod_{n'\in \mathbb{supp}(\vec{t})\setminus \{n\}}{m_{n'} \choose t_{n'}}\right]={K-1 \choose t},
\end{align}
since the left hand side is the number of ways to choose $t+1$ users among the $K$ users, with $\ell^{[n]}$ already chosen, counted one transmission type at a time. It follows that for any $\vec{d}\in \mathcal{D}$, 
\begin{align}
R_{\vec{d},\vec{\mathcal{P}}^{(t)}_{\vec{d}}}=N^*{K-1 \choose t}. \label{eqn:TianRs}
\end{align}
Let us turn to (\ref{eqn:Ms}), the second term of which in this case can be simplified as
\begin{align}
&\Delta M_{\vec{d}, \vec{\mathcal{P}}^{(t)}_{\vec{d}}, k} = \sum_{\vec{t} \in \mathcal{T}^{(t)}_{\vec{d}}: t_{d_k} >0} \sum_{n \in \mathbb{supp}(\vec{t})\setminus\{d_k\}} \left(\binom{m_{d_k}-1}{t_{d_k}-1} \binom{m_n-1}{t_n-1}\cdot \prod_{n' \in \mathbb{supp}(\vec{t})\setminus\{n,d_k\}} \binom{m_{n'}}{t_{n'}} \right)\nonumber\\
&=\sum_{\vec{t} \in \mathcal{T}^{(t)}_{\vec{d}}: t_{d_k} >0} \sum_{n \in \mathbb{supp}(\vec{m})\setminus\{d_k\}} \mathbb{1}_{n\in \mathbb{supp}(\vec{t})}
\left[\binom{m_{d_k}-1}{t_{d_k}-1}
 \binom{m_n-1}{t_n-1}\cdot \prod_{n' \in \mathbb{supp}(\vec{t})\setminus\{n,d_k\}} \binom{m_{n'}}{t_{n'}}\right]\nonumber\\
 &=\sum_{n \in \mathbb{supp}(\vec{m})\setminus \{d_k\}} \sum_{\vec{t} \in \mathcal{T}^{(t)}_{\vec{d}}} \mathbb{1}_{\{n,d_k\}\subseteq \mathbb{supp}(\vec{t})}
\left[
 \binom{m_{d_k}-1}{t_{d_k}-1} \cdot   \binom{m_n-1}{t_n-1}\cdot \prod_{n' \in \mathbb{supp}(\vec{t})\setminus\{n,d_k\}} \binom{m_{n'}}{t_{n'}} \right]\nonumber\\
 &=(N^*-1){K-2 \choose t-1},
\end{align}
where the last equality is because for each fixed $n\in \mathbb{supp}(\vec{m})\setminus\{d_k\}$, the inner summation is simply the number of ways to choose $t+1$ users in the $K$ users, with $\ell^{[d_k]}$ and $\ell^{[n]}$ already chosen. Thus we arrive at
\begin{align}
M_{\vec{d}, \vec{\mathcal{P}}^{(t)}_{\vec{d}}, k}=N \binom{K-1}{t-1}-(N^*-1){K-2 \choose t-1}. \label{eqn:MTian}
\end{align}
Note that neither (\ref{eqn:TianRs}) nor (\ref{eqn:MTian}) depends on $\vec{d}$ or $k$.

For $N^*=\tilde{N}$, normalizing both of them by ${K \choose t}$ already gives exactly the memory-rate tradeoff pairs in (\ref{eqn:TianChen}). This leaves us only the case when $N^*\neq \tilde{N}$ to consider.  We shall use two decomposition patterns for this case. Define
\begin{align}
\alpha_{\breve{\vec{\mathcal{P}}}_{\vec{d}}^{(t)}}=\left\{
\begin{array}{ll}
\frac{\tilde{N}-N^*}{K-N^*}& K-t\leq \tilde{N}\\
\frac{(K-t)(\tilde{N}-N^*)}{K(\tilde{N}-N^*)+tN^*}&\text{otherwise}
\end{array}
\right.
\end{align}
which is clearly non-negative and is associated with the uncoded transmission pattern, and $1-\alpha_{\breve{\vec{\mathcal{P}}}_{\vec{d}}^{(t)}}$ which is  also non-negative and is associated with the transmission pattern whose rate and memory are given in (\ref{eqn:TianRs}) and (\ref{eqn:MTian}). It is easy to check that when $K-t\leq \tilde{N}$, 
\begin{align}
&\alpha_{\breve{\vec{\mathcal{P}}}_{\vec{d}}^{(t)}}M_{\vec{d}, \breve{\vec{\mathcal{P}}}^{(t)}_{\vec{d}}, k}+(1-\alpha_{\breve{\vec{\mathcal{P}}}_{\vec{d}}^{(t)}})M_{\vec{d}, \vec{\mathcal{P}}^{(t)}_{\vec{d}}, k}=N{K-1 \choose t-1}-(\tilde{N}-1){K-2 \choose t-1}.
\end{align}
and 
\begin{align}
\alpha_{\breve{\vec{\mathcal{P}}}_{\vec{d}}^{(t)}}R_{\vec{d}, \breve{\vec{\mathcal{P}}}^{(t)}_{\vec{d}}, k}+(1-\alpha_{\breve{\vec{\mathcal{P}}}_{\vec{d}}^{(t)}})R_{\vec{d}, \vec{\mathcal{P}}^{(t)}_{\vec{d}}, k}=\tilde{N}{K-1 \choose t}. \label{eqn:Requal}
\end{align}
On the other hand, when $K-t> \tilde{N}$, (\ref{eqn:Requal}) still holds, but 
\begin{align}
\alpha_{\breve{\vec{\mathcal{P}}}_{\vec{d}}^{(t)}}M_{\vec{d}, \breve{\vec{\mathcal{P}}}^{(t)}_{\vec{d}}, k}+(1-\alpha_{\breve{\vec{\mathcal{P}}}_{\vec{d}}^{(t)}})M_{\vec{d}, \vec{\mathcal{P}}^{(t)}_{\vec{d}}, k}&=N{K-1 \choose t-1}-{K-2 \choose t-1}\left(\tilde{N}-\frac{\tilde{N}(\tilde{N}-N^*)+t\tilde{N}}{K(\tilde{N}-N^*)+tN^*}\right)\nonumber\\
&< N{K-1 \choose t-1}-{K-2 \choose t-1}(\tilde{N}-1),
\end{align}
where the last inequality is because
\begin{align}
\tilde{N}(\tilde{N}-N^*)+t\tilde{N}< (K-t)(\tilde{N}-N^*)+t\tilde{N}=K(\tilde{N}-N^*)+tN^*
\end{align}
by the condition $K-t> \tilde{N}$. Thus for $N^*\neq \tilde{N}$, we have found the correct $\alpha_{\breve{\vec{\mathcal{P}}}_{\vec{d}}^{(t)}}$ to satisfy the conditions in (\ref{eqn:LP1})-(\ref{eqn:LP5}), and indeed the memory-rate pair (\ref{eqn:TianChen}) is in the region $\mathcal{R}^{(t)}$. The proof is thus complete.

\bibliographystyle{IEEEtran}
 \newcommand{\noop}[1]{}

\end{document}